\documentclass[paper=a4, abstract=true, fontsize=11pt]{scrartcl}
\usepackage[T1]{fontenc}
\usepackage{fourier}
\usepackage[english]{babel}								 
\usepackage{amssymb}
\usepackage[protrusion=true,expansion=true]{microtype}	
\usepackage{amsmath,amsfonts,amsthm} 
\usepackage[pdftex]{graphicx}	
\usepackage{url}
\usepackage[title,titletoc,page]{appendix} 
\usepackage{listings} 
\usepackage{color}
\usepackage{cite}
\usepackage{verbatim}
\usepackage{sectsty}
\usepackage{caption}
\usepackage{subcaption}
\usepackage[ruled,vlined]{algorithm2e}
\usepackage{diagbox}
\usepackage{tabularx}
\usepackage{float}
\usepackage{booktabs}

\allsectionsfont{\centering \normalfont\scshape}

\usepackage{fancyhdr}
\numberwithin{equation}{section}		
\numberwithin{figure}{section}			
\numberwithin{table}{section}				

\newtheorem{thm}{Theorem}[section]  
\newtheorem{lem}[thm]{Lemma}

\newtheorem{prop}[thm]{Proposition} 
 
\newtheorem*{remark}{Remark}

\newcommand{\bparenths}[1]{\Big({#1}\Big)}

\newcommand{\vect}[1]{\boldsymbol{\mathbf{#1}}}
\newcommand{\leqs}{\leqslant}
\newcommand{\geqs}{\geqslant}

\newcommand{\babs}[1]{\left|{#1}\right|}
\newcommand{\bnorm}[1]{\left|\left|{#1}\right|\right|}
\newcommand{\btwonorm}[1]{\left|\left|{#1}\right|\right|_2}

\title{Mathematical Foundation of Sparsity-based Multi-snapshot Spectral Estimation \thanks{\footnotesize This work was supported in part by the Swiss National Science Foundation grant number
200021--200307.}}
\author{Ping Liu\thanks{\footnotesize Department of Mathematics, ETH Z\"urich, R\"amistrasse 101, CH-8092 Z\"urich, Switzerland (ping.liu@sam.math.ethz.ch, habib.ammari@math.ethz.ch).} \and Sanghyeon Yu\thanks{\footnotesize Department of Mathematics, Korea University, Seoul 02841, S. Korea (sanghyeon\_yu@korea.ac.kr).}\and Ola Sabet\thanks{\footnotesize Department of Molecular Life Sciences, University of Z\"urich, Switzerland (ola.sabet@mls.uzh.ch, lucas.pelkmans@mls.uzh.ch).}  \and Lucas Pelkmans\footnotemark[4] \and Habib Ammari\footnotemark[2]}

\begin{document}
\date{}	
\maketitle	

\begin{abstract}
In this paper, we study the spectral estimation problem of estimating the locations of a fixed number of point sources given multiple snapshots of Fourier measurements in a bounded domain. We aim to provide a mathematical foundation for sparsity-based super-resolution in such spectral estimation problems in both one- and multi-dimensional spaces. In particular, we estimate the resolution and stability of the location recovery of a cluster of closely spaced point sources when considering the sparsest solution under the measurement constraint, and characterize their dependence on the cut-off frequency, the noise level, the sparsity of point sources, and the incoherence of the amplitude vectors of point sources. Our estimate emphasizes the importance of the high incoherence of amplitude vectors in enhancing the resolution of multi-snapshot spectral estimation. Moreover, to the best of our knowledge, it also provides the first stability result in the super-resolution regime for the well-known sparse MMV problem in DOA estimation.




\end{abstract}

\vspace{0.5cm}
		
\vspace{0.2cm}
		
	\vspace{0.5cm}

\section{Introduction}
In this paper, we study the spectral estimation problem of estimating the locations of a fixed number of point sources given multiple snapshots of Fourier measurements in a bounded domain. This problem appears in many imaging
and signal processing applications, including biological imaging \cite{STED, PALM, STORM,SIM, blindSIM, beam2022}, inverse scattering \cite{fannjiang2011music}, Direction-Of-Arrival (DOA) estimation \cite{krim1996two, ottersten1993exact}, and spectral analysis \cite{stoica2005spectral}. A particular core of such problems in spectral estimation is named super-resolution, which seeks to reconstruct point sources separated less than the Rayleigh limit \cite{rayleigh1879xxxi}. Over the past few years, the capability of super-resolution in single-snapshot spectral estimation has already been established by several mathematical theories \cite{liu2021mathematicaloned, liu2021mathematicalhighd, liu2021theorylse,  batenkov2019super, li2021stable, demanet2015recoverability, liu2023super} and the resolution limits have been explicitly characterized \cite{liu2021theorylse, liu2022mathematicalSR, liu2021mathematicalhighd, liu2022mathematicalpositive, liu2023improved}, but the capabilities of super-resolution from multiple snapshots are not well understood. In a recent study \cite{liu2023operator}, the resolution for the multi-illumination imaging with known (or well-approximated) illumination patterns was analyzed, encompassing the case of resolving both general and sparse sources. But in many applications like DOA estimation, the "illumination patterns" (or amplitude vectors) are entirely unknown, which the aforementioned theory cannot address. We analyze such cases in this paper, aiming to lay a mathematical foundation for sparsity-based multi-snapshot spectral estimation.

\subsection{Our contribution and related work}
Our first main result establishes the theory for the resolution and stability of one-dimensional multi-snapshot spectral estimation. In particular, we consider the collection of point sources in each snapshot as
\begin{equation}\label{equ:contribution1}
\mu_t = \sum_{j=1}^n a_{j, t} \delta_{y_j}, \quad t=1,\cdots, T,
\end{equation}
where $n$ is the number of point sources, $y_j$'s are source locations, and $a_{j,t}$'s are the corresponding amplitudes at $t$-th snapshot. As we focus on the super-resolving capability of the spectral estimation in this paper, we consider throughout that $y_j$'s are clustered and tightly spaced within a connected region $\mathcal O$, which has a length (or radius) spanning several Rayleigh limits. The available measurements are the noisy Fourier data of  $\mu_t$ in a bounded interval $[-\Omega, \Omega]$ given by
\begin{equation}
\mathbf Y_t(\omega) = \mathcal F [\mu_t] (\omega) + \mathbf W_t(\omega)= \sum_{j=1}^{n}a_{j, t} e^{i y_j \omega} + \mathbf W_t(\omega), \ 1\leqs t\leqs T, \ \omega \in [-\Omega, \Omega], 
\end{equation}
where $\mathcal F[\cdot ]$ denotes the Fourier transform and $\vect W_t$ is the noise satisfying $\babs{\vect W_t(\omega)}<\sigma, 1\leqs t\leqs T,  \omega \in [-\Omega, \Omega]$. We reconstruct the source locations by the $l_0$-minimization problem:
\begin{equation}\label{equ:contribution2}
\begin{aligned}
&\min_{\{\widehat y_1, \cdots, \widehat y_k\}, \widehat y_j \in \mathcal O, j=1,\cdots, k}\left(\#\{\widehat y_1, \cdots, \widehat y_k\}\right) \\
&\text{subject to the existence of $\widehat a_{j,t}$'s such that}\  \bnorm{\mathcal F\left[\sum_{j=1}^k\widehat a_{j,t} \delta_{\widehat y_j}\right] -\vect Y_t}_{\infty}< \sigma, \quad 1\leqs t\leqs T,
\end{aligned}
\end{equation} 
where $\bnorm{f}_\infty :=\max_{\omega \in [-\Omega, \Omega]}|f(\omega)|$ for a function $f$.

Informally, our main result (Theorem \ref{thm:l0normrecovery0}) states the following.  Suppose that the $n$ point sources  $\delta_{y_j}$'s in (\ref{equ:contribution1}) satisfy 
\begin{equation}\label{equ:sepaconditionl0normrecovery1}
d_{\min} := \min_{p\neq j}\babs{y_p-y_j}\geqs \frac{2.4c_0e\pi }{\Omega }\Big(\frac{\sigma}{\sigma_{\infty, \min}(\vect A^{\top})}\Big)^{\frac{1}{n}},
\end{equation}
where $c_0$ corresponds to the length of $\mathcal O$, $\vect A: = (a_{j,t})_{1\leqs j\leqs n, 1\leqs t \leqs T}$ is the amplitude matrix, and $\sigma_{\infty, \min}\left(\vect A^{\top}\right)$ characterizes the correlation between the amplitude vectors of different sources. Then any solution to the $l_0$ optimization problem (\ref{equ:contribution2}) contains exactly $n$ points. Moreover, for $\{\widehat y_1, \cdots, \widehat y_n\}$ being a corresponding solution, after reordering $\widehat y_j$'s, we have $\left|\widehat y_j-y_j\right|<\frac{d_{\min}}{2}$
	and 
	\begin{equation}\label{equ:stabilityl0normrecovery1}
 \Big|\widehat y_j-y_j\Big|< \frac{C(n)}{\Omega}\mathrm{SRF}^{n-1}\frac{\sigma}{\sigma_{\infty, \min}\left(\vect A^{\top}\right)}, \quad 1\leqs j\leqs n,
	\end{equation}
where $C(n)=\sqrt{\pi}n^2c_0^{n}e^{n}$ and $\mathrm{SRF} = \frac{\pi}{\Omega d_{\min}}$ is the super-resolution factor. In particular, (\ref{equ:sepaconditionl0normrecovery1}) estimates the resolution for the sparse recovery in the super-resolution problem and (\ref{equ:stabilityl0normrecovery1}) estimates the stability of the location recovery in such a scenario.  Estimates (\ref{equ:sepaconditionl0normrecovery1}) and (\ref{equ:stabilityl0normrecovery1}) emphasize the importance of the sparsity of sources, cut-off frequency, noise level, and particularly the high incoherence of amplitude vectors in enhancing the resolution and stability of multi-snapshot spectral estimation. These estimates are then generalized to the multi-dimensional case in Theorem \ref{thm:highdl0normrecovery0}. Moreover,  Theorem \ref{thm:discretel0normrecovery0} establishes such estimates for a discrete model, giving a stability estimate for the sparse multiple measurement vectors (MMV) problem \cite{cotter2005sparse, chen2006theoretical, davies2012rank} in DOA estimation. 

   The most relevant work to our paper is \cite{li2022stability}, in which the authors analyzed the stability and super-resolution of MUSIC \cite{schmidt1986multiple} and ESPRIT \cite{roy1989esprit} for multi-snapshot spectral estimation and derived estimates akin to ours. The MUSIC and ESPRIT methods are not directly sparsity-based methods. They reconstruct the source locations through subspace-based techniques and are usually called subspace methods. Under assumptions on the randomness of the noise, the authors estimated in \cite{li2022stability} respectively the perturbation of the noise-space correlation function (NSC) for the MUSIC algorithm and the location recovery stability for the ESPRIT algorithm. Leveraging the notation in the current paper, their estimate for the output $\widehat y_j$'s of the ESPRIT algorithm reads:
   
\noindent a) Moderate SNR regime. If 
\[
\sigma \lesssim \sigma_n(\vect \Phi)\sqrt{\lambda(X)},
\]
then the averaged location recovery error satisfies 
\begin{equation}\label{equ:esprit1}
\mathbb{E}\left(\min _\psi \max _j\left|\widehat{y}_{\psi(j)}-y_j\right|_{\mathbb{T}}\right) \lesssim \frac{\Omega \sigma}{\sqrt{\lambda(X)} \sigma_n^2(\vect \Phi)} , 
\end{equation}
where $\psi$ is a permutation on $\{1, \ldots, n\}$, $\babs{\cdot}_{\mathbb T}$ is a certain wrap-around distance, $\sigma_{n}(\vect \Phi)$ denotes the $n$-th largest singular value of the Vandermonde matrix $\vect \Phi$ that is similar to (\ref{equ:defineofphi1}) and generated by $y_j$'s, and $\lambda(X)$ is the minimum eigenvalue of the amplitude covariance matrix,
\[
X:= \sum_{t=1}^{T}\alpha_t \alpha_t^*,
\]
with $\alpha_t:=(a_{1,t}, \cdots, a_{n,t})^{\top}$ being the amplitude vector in the $t$-th snapshot. 

\noindent b) Large SNR regime. If
$$
\sigma\lesssim \frac{\sigma_n(\vect \Phi)\sqrt{\lambda(X)}}{\sqrt{\Omega}} 
\min\left(1, \frac{\sigma_{n}(\vect \Phi)\sqrt{d_{\min}}}{\sqrt{\Omega}}, \frac{\sigma_{n}(\vect \Phi)^2d_{\min}}{\Omega}\right),
$$
then it holds that
\begin{equation}\label{equ:esprit2}
\mathbb{E}\left(\min _\psi \max _j\left|\widehat{y}_{\psi(j)}-y_j\right|_{\mathbb{T}}\right)\lesssim \frac{\sqrt{\Omega}\sigma}{\sqrt{\lambda(X)} \sigma_n(\vect \Phi)} .
\end{equation}
In particular, it is well-known \cite[Theorem 2.7]{li2021stable} that for $n$ point sources with minimum separation $d_{\min}$ in a single cluster, $\sigma_{n}(\vect \Phi)\gtrsim \sqrt{\Omega}(\Omega d_{\min})^{n-1}$, whence (\ref{equ:esprit2}) becomes 
\begin{equation}\label{equ:esprit3}
\mathbb{E}\left(\min _\psi \max _j\left|\widehat{y}_{\psi(j)}-y_j\right|_{\mathbb{T}}\right)\lesssim \mathrm{SRF}^{n-1}\frac{\sigma}{\sqrt{\lambda(X)}},
\end{equation}
which is nearly the same as our estimate (\ref{equ:stabilityl0normrecovery1}), except for the missing of cut-off frequency $\Omega$. Note that, as $\sqrt{\lambda(x)}=\sqrt{\lambda\left(\vect A \vect A^*\right)}=\sigma_{\min}(\vect A^*)$,  $\sigma_{\infty, \min}\left(\vect A^{\top}\right)$ and $\sqrt{\lambda(x)}$ are comparable; see (\ref{equ:sigmainftyminestimate1}) and (\ref{equ:gaussianillum1}). The drawback of the stability estimate (\ref{equ:esprit3}) is the noise level requirement
\[
\sigma \lesssim \frac{\sigma_{n}(\vect \Phi)^3\sqrt{\lambda(X)}d_{\min}}{\sqrt{\Omega}\Omega}, 
\]
 which is far more restrictive. In comparison, as seen from (\ref{equ:sepaconditionl0normrecovery1}), our requirement is only
 \[
 \sigma \lesssim (\Omega d_{\min})^{n}\sigma_{\infty, \min}\left(\vect A^{\top}\right),
 \]
 which is consistent with the stability estimate (\ref{equ:stabilityl0normrecovery1}). We also conjecture that this requirement is already optimal to ensure a stable location reconstruction. On the other hand, our estimate (\ref{equ:stabilityl0normrecovery1}) also holds in the moderate SNR regime, giving a better estimate than (\ref{equ:esprit1}). We also notice a recent work \cite{yang2022nonasymptotic} on the nonasymptotic performance analysis of ESPRIT, obtaining an estimate similar to (\ref{equ:esprit1}).


It is worth emphasizing that there were already many mathematical theories for estimating the stability of super-resolution in the single measurement case. To the best of our knowledge, the first work for the stability of super-resolving multiple sources was by Donoho \cite{donoho1992superresolution}. He considered a grid setting where a discrete measure is supported on a lattice (spacing by $\Delta$) and regularized by a so-called  "Rayleigh index". His main contribution is estimating the corresponding minimax error in the recovery, which emphasizes the importance of the sparsity of sources for super-resolution. It was improved in recent years for the case when super-resolving $n$-sparse on-the-grid sources \cite{demanet2015recoverability}, where the minimax error of amplitude recovery was shown to scale as $\mathrm{SRF}^{2 n-1} \sigma$ with $\mathrm{SRF}:=\frac{1}{\Delta \Omega}$ being the super-resolution factor. The case of resolving multi-cluster sources was considered in \cite{li2021stable, batenkov2020conditioning} and similar minimax error estimations were established. 

In \cite{akinshin2015accuracy, batenkov2019super}, the authors considered the minimax error for recovering off-the-grid point sources. They showed that for $\sigma \lesssim (\mathrm{SRF})^{-2p+1}$, where $p$ is the number of point sources in a cluster, the minimax error for the amplitude and the location recoveries scale respectively as $(\mathrm{SRF})^{2p-1}\sigma$ and $(\mathrm{SRF})^{2p-2} {\sigma}/{\Omega}$. Moreover, for the isolated non-cluster point sources, the corresponding minimax error for the amplitude and the location recoveries scale respectively as $\sigma$ and ${\sigma}/{\Omega}$. These estimates were generalized to the case of resolving positive sources in \cite{liu2023super} recently. 

On the other hand, to analyze the resolution for recovering multiple point sources, in \cite{liu2021mathematicaloned, liu2021mathematicalhighd, liu2021theorylse, liu2023improved} the authors defined "computational resolution limits" which characterize the minimum required distance between point sources so that their number and locations can be stably resolved under certain noise level. By developing a non-linear approximation theory in a so-called Vandermonde space, they derived bounds for computational resolution limits for a deconvolution problem \cite{liu2021mathematicaloned} and a spectral estimation problem \cite{liu2021theorylse}. In particular, they showed in \cite{liu2021theorylse} that the computational resolution limit for number and location recovery should be respectively $\frac{C_{\mathrm{num}}}{\Omega}\left(\frac{\sigma}{m_{\min}}\right)^{\frac{1}{2n-2}}$ and  $\frac{C_{\mathrm{supp}}}{\Omega}\left(\frac{\sigma}{m_{\min}}\right)^{\frac{1}{2n-1}}$, where $C_{\mathrm{num}}$ and $C_{\mathrm{supp}}$ are constants and $m_{\min}$ is the minimum magnitude of the source amplitude. Their results demonstrate that when the point sources are separated larger than $\frac{C_{\mathrm{supp}}}{\Omega}\left(\frac{\sigma}{m_{\min}}\right)^{\frac{1}{2n-1}}$, we can stably recover the source locations. Conversely, when the point sources are separated by a distance less than $O\left(\frac{C_{\mathrm{supp}}}{\Omega}\left(\frac{\sigma}{m_{\min}}\right)^{\frac{1}{2n-1}}\right)$, stably recovering the source locations is impossible in the worst case. This resolution limit indicates that super-resolution is possible for the single measurement case but requires a very high signal-to-noise ratio (according to the exponent $\frac{1}{2n-1}$). These estimates were later generalized to the problem of resolving positive sources \cite{liu2022mathematicalpositive} and moving sources \cite{liu2023dynamic}. As we have seen, the mathematics for the super-resolution from a single snapshot was well-established. Nevertheless, the multi-snapshot super-resolution still lacks or even is without any
mathematical foundation, despite its prominent importance. This strongly motivates the current paper.


For the super-resolution algorithms, the subspace methods such as MUSIC \cite{schmidt1986multiple, stoica1989music, li2021stable, li2022stability}, ESPRIT \cite{roy1989esprit, li2020super}, and Matrix Pencil Method \cite{hua1990matrix} are widely used in applications due to their excellent performance. They can be dated back to Prony’s method \cite{Prony-1795}. In recent years, the rise in popularity of sparse modeling and compressive sensing has led to the creation of numerous sparsity-promoting super-resolution algorithms. In the groundbreaking work of Cand\`es and Fernandez-Granda \cite{candes2014towards}, it was demonstrated that off-the-grid sources can be exactly recovered from their low-pass Fourier coefficients by total variation minimization under a minimum separation condition. Other well-known sparsity-promoting methods include the BLASSO algorithm \cite{azais2015spike, duval2015exact, poon2019} and the atomic norm minimization method \cite{tang2013compressed, tang2014near}. These two algorithms were proved to be able to stably recover the sources under a minimum separation condition or a non-degeneracy condition. The resolution of these convex algorithms is limited by a distance on the scale of the Rayleigh limit \cite{tang2015resolution, da2020stable} for recovering generic signed point sources. When super-resolving positive sources \cite{morgenshtern2016super, morgenshtern2020super, bendory2017robust, denoyelle2017support}, such constraint on resolution can be eased and the algorithm performance is nearly-optimal. However, most of the sparsity-promoting algorithms are designed for single-snapshot super-resolution, with scarce stability estimates for multi-snapshot equivalents. In a recent paper on a new multi-snapshot super-resolution algorithm \cite{fei2023iff}, the authors estimated a resolution characterizing when the focused source (a single point) was closed to one of the underlying sources to provide the theory for partial steps of their method. In direct comparison, the focus of our paper diverges considerably as we consider the recovery of all point sources. 

\subsection{Organization of the paper}
Our paper is organized in the following way.
Section \ref{sect2} formulates the minimization problem for recovering point sources from multiple snapshots. Sections \ref{sect3} and \ref{section:highdcase} present the main results in respectively the one- and multi-dimensional case and a detailed discussion on their significance. 
Section \ref{section:approxinvandermonde}  introduces the main technique (namely the approximation theory in Vandermonde space) that is used to show the main results of this paper. 
In Section \ref{section:proofofthml0normrecover}, Theorems \ref{thm:l0normrecovery0} and \ref{thm:discretel0normrecovery0}
are proved. Section \ref{section:proofhighdl0normrecovery} is devoted to the proof of Theorem \ref{thm:highdl0normrecovery0}.
Finally, the appendix provides some lemmas and inequalities
that are used in the paper.

\section{One-dimensional case} \label{sect3}

\subsection{Problem setting} \label{sect2}
To facilitate reading, we reintroduce the model setup for the one-dimensional multi-snapshot spectral estimation. Let $n$ be the number of point sources and $y_j\in \mathbb R, j=1,\cdots n,$ be the corresponding locations. The point sources are denoted by $\delta_{y_j}$'s with $\delta$ being the Dirac measure. We denote the source amplitude of $\delta_{y_j}$ at $t$-th snapshot by $a_{j,t}$. Thus the collection of point sources in each snapshot writes 
\begin{equation}\label{equ:sources1}
\mu_t = \sum_{j=1}^n a_{j, t} \delta_{y_j}, \quad t=1,\cdots, T,
\end{equation}
where $T$ is the total number of snapshots. The available measurements (snapshots) are the noisy Fourier data of  $\mu_t$ in a bounded interval. More precisely, they are given by 
\begin{equation}\label{equ:multimodelsetting1}
\mathbf Y_t(\omega) = \mathcal F [\mu_t] (\omega) + \mathbf W_t(\omega)= \sum_{j=1}^{n}a_{j, t} e^{i y_j \omega} + \mathbf W_t(\omega), \ 1\leqs t\leqs T, \ \omega \in [-\Omega, \Omega], 
\end{equation}
where $\mathcal F[\cdot]$ denotes the Fourier transform and $\vect W_t(\omega)$ is the noise. Here, $\Omega$ is called the cut-off frequency, it represents the cut-off frequency in the imaging problem or the boundary of the sensor array in the DOA estimation. We assume 
\begin{equation}\label{equ:noiseconstraint1}
\babs{\vect W_t(\omega)}<\sigma, \quad 1\leqs t\leqs T, \ \omega \in [-\Omega, \Omega],
\end{equation}
with $\sigma$ being the noise level. 

We suppose that the point sources are located in an interval $\mathcal O\subset \mathbb R$ with a length of several Rayleigh resolution limits. The inverse problem we are concerned with is to recover the sparsest location set $\{\widehat y_1, \cdots, \widehat y_k\}\subset \mathcal O$ that could generate these snapshots $\vect Y_t$'s. In particular, we consider the following $l_0$-minimization problem:
\begin{equation}\label{prob:l0minimization2}
\begin{aligned}
&\min_{\{\widehat y_1, \cdots, \widehat y_k\} \subset \mathcal O}\left(\#\{\widehat y_1, \cdots, \widehat y_k\}\right) \\
&\text{subject to the existence of $\widehat a_{j,t}$'s such that}\  \bnorm{\mathcal F\left[\sum_{j=1}^k\widehat a_{j,t} \delta_{\widehat y_j}\right] -\vect Y_t}_{\infty}< \sigma, \quad 1\leqs t\leqs T,
\end{aligned}
\end{equation} 
where $\bnorm{f}_\infty :=\max_{\omega \in [-\Omega, \Omega]}|f(\omega)|$ for a function $f$.

Our main result in the next section gives an estimation of the resolution and stability of the sparse recovery problem (\ref{prob:l0minimization2}) in dimension one.


\subsection{Main results}
Let us first introduce some notation. We define the amplitude matrix as
\begin{equation}\label{equ:illuminationpattern1}
\vect A := \begin{pmatrix}
a_{1,1}&\cdots&a_{1,T}\\
\vdots&\vdots&\vdots\\
a_{n,1}&\cdots&a_{n,T}\\
\end{pmatrix}, 
\end{equation}
where $a_{j,t}$'s are the source amplitudes in (\ref{equ:sources1}) and the vector $(a_{j, 1}, \cdots, a_{j, T})^{\top}$ is called the amplitude vector of point source $\delta_{y_j}$. For a $q\times k$ matrix $B$, we define $\sigma_{\infty, \min}(B)$ as
\begin{equation} \label{sigmadef}
\sigma_{\infty, \min}(B) : = \min_{x\in \mathbb C^k, \mathbb ||x||_{\infty}\geqs 1} ||Bx||_{\infty},
\end{equation}
characterizing the correlation between the columns of $B$. Throughout this paper, for a complex matrix $B$, we denote $B^\top$ its transpose and $B^*$ its conjugate transpose. 

We have the following result on the stability of the problem (\ref{prob:l0minimization2}), whose proof is given in Section \ref{section:proofofthml0normrecover}.
\begin{thm}\label{thm:l0normrecovery0}
Let $n\geqs 2$ and let the interval $\mathcal O\subset \mathbb R$ be of length  $\frac{c_0n\pi}{2\Omega}$ with $c_0\geqs 1$. Let $\vect Y_t$'s in (\ref{equ:multimodelsetting1}) be the measurements generated by $n$ point sources at $\{y_1, \cdots, y_n\}, y_j \in \mathcal O, j=1,\cdots, n$. Suppose that the following separation condition holds:
 \begin{equation}\label{equ:sepaconditionl0normrecovery}
	d_{\min} := \min_{p\neq j}\babs{y_p-y_j}\geqs \frac{2.4c_0e\pi }{\Omega }\Big(\frac{\sigma}{\sigma_{\infty, \min}(\vect A^{\top})}\Big)^{\frac{1}{n}},
	\end{equation}
	with  $\frac{\sigma}{\sigma_{\infty, \min}(\vect A^{\top})}\leqs 1$. Then, any solution to (\ref{prob:l0minimization2}) contains exactly $n$ points. Moreover, for $\{\widehat y_1, \cdots, \widehat y_n\}$ being a corresponding solution, after reordering the $\widehat y_j$'s, we have 
	\begin{equation}
	\Big|\widehat y_j-y_j\Big|<\frac{d_{\min}}{2},
	\end{equation} 
	and 
	\begin{equation}
 \Big|\widehat y_j-y_j\Big|< \frac{C(n)}{\Omega}\mathrm{SRF}^{n-1}\frac{\sigma}{\sigma_{\infty, \min}\left(\vect A^{\top}\right)}, \quad 1\leqs j\leqs n,
	\end{equation}
	where $C(n)=\sqrt{\pi}n^2c_0^{n}e^{n}$ and $\mathrm{SRF} = \frac{\pi}{\Omega d_{\min}}$ is the super-resolution factor.
\end{thm}


\begin{remark}
Note that the stability result in Theorem \ref{thm:l0normrecovery0} holds for any algorithm that can recover the sparsest solution (solution with $n$ point sources) satisfying the measurement constraint. 
\end{remark}


Theorem \ref{thm:l0normrecovery0} demonstrates that when the point sources are separated by the distance $d_{\min}$ in (\ref{equ:sepaconditionl0normrecovery}), each of the recovered locations from (\ref{prob:l0minimization2}) lies in a neighborhood of the ground truth, with an estimated deviation also provided. Based on formula (\ref{equ:sepaconditionl0normrecovery}), we demonstrate that the incoherence (encoded in $\sigma_{\infty, \min}\left(\vect A^{\top}\right)$) between the amplitude vectors is crucial to the sparsity-based super-resolution. Recall that $\frac{C}{\Omega}$ for a constant $C\approx \pi$ is the so-called Rayleigh resolution limit. The estimate (\ref{equ:sepaconditionl0normrecovery}) also indicates that super-resolution from multiple snapshots is possible provided a sufficiently small noise level and large enough $\sigma_{\infty, \min}
\left(\vect A^{\top}\right)$. In particular, the resolution for the multi-snapshot spectral estimation with highly uncorrelated amplitude vectors is significantly better than the one for single-snapshot spectral estimation; see Section \ref{section:comparisonsinglesnapshot}.

For other detailed discussions on the result, we leave it to subsequent subsections.




\subsection{Properties of $\sigma_{\infty, \min}\left(\vect A^{\top}\right)$}

\textbf{Adding same signals will not enhance the resolution:}\\
Let 
\[
\vect A = \begin{pmatrix}
a_{1,1}&\cdots&a_{1,T}\\
\vdots&\vdots&\vdots\\
a_{n,1}& \cdots & a_{n, T}
\end{pmatrix}, \quad \widehat {\mathbf A} = \begin{pmatrix}
a_{1,1}&\cdots&a_{1,T}& a_{1, T+1}\\
\vdots&\vdots&\vdots\\
a_{n,1}&\cdots&a_{n,T} & a_{n, T+1}\\
\end{pmatrix}
\]
with $(a_{1,T},\cdots, a_{n,T})^{\top
} = (a_{1,T+1},\cdots, a_{n,T+1})^{\top
}$. By the definition of $\sigma_{\infty, \min}(\cdot)$, it is clear that $\sigma_{\infty, \min}(\widehat {\mathbf A}^{\top}) =\sigma_{\infty, \min}(\vect A^{\top})$. Thus, adding the same signal cannot increase the resolution in Theorem \ref{thm:l0normrecovery0}. 

\medskip
\noindent \textbf{The incoherence is crucial:}\\
The value of $\sigma_{\infty, \min}\left(\vect A^{\top}\right)$ is related to the correlation between the columns of the matrix $\vect A^{\top}$. In particular, we have the following rough estimation of  $\sigma_{\infty, \min}\left(\vect A^{\top}\right)$:
\begin{equation}\label{equ:sigmainftyminestimate1}
\sigma_{\infty, \min}\left(\vect A^{\top}\right)\geqs \frac{\sigma_{\min}\left(\vect A^{\top}\right)}{\sqrt{T}},
\end{equation}
where $\sigma_{\min}\left(\vect A^{\top}\right)$ is the minimum singular value of $\vect A^{\top}$. This clearly illustrates that the correlation between the columns of $\vect A^{\top}$ (i.e., the amplitude vectors of the sources) is crucial to $\sigma_{\infty, \min}\left(\vect A^{\top}\right)$. 

In particular, suppose that the amplitudes are independent Gaussian variables, i.e., $a_{j,t}\sim \mathcal N(0,1)$. We have 
\begin{equation}\label{equ:gaussianillum1}
\sigma_{\infty, \min}\left(\vect A^{\top}\right)\rightarrow 1
\end{equation}
as $T\rightarrow \infty$, which can be seen from (\ref{equ:sigmainftyminestimate1}) and the fact that the minimum eigenvalue of $\left(\frac{1}{T}\vect A \vect A^{\top}\right)$ tends to $1$.

\subsection{Comparison with the single snapshot case}\label{section:comparisonsinglesnapshot}
In this subsection, we compare the resolution in the single snapshot case with that in the multiple snapshots case, whereby we illustrate the effect of multiple snapshots in enhancing the resolution. 

In \cite{liu2021theorylse}, the authors estimate the so-called computational resolution limit for the line spectral estimation problem (or DOA estimation) of the single measurement case. The results in \cite{liu2021theorylse} show that, for the single measurement case, when the point sources are separated by
\[
\tau = \frac{c}{\Omega} \sigma^{\frac{1}{2n-1}},
\]
for some positive constant $c$,  there exists a discrete measure $\mu = \sum_{j=1}^n a_j \delta_{y_j}$ with $n$ point sources located at $\{-\tau, -2\tau, -n\tau\}$ and another discrete measure $\widehat \mu = \sum_{j=1}^n \widehat a_j \delta_{\widehat y_j}$ with $n$ point sources located at  $\{0,\tau,\cdots, (n-1)\tau\}$ such that
\[
\bnorm{\mathcal F[\widehat \mu]-\mathcal F[\mu]}_{\infty}< \sigma,
\]
and the minimum magnitude of amplitudes of $\widehat \mu$ and $\mu$ are of order one. 

This result demonstrates that when the point sources are separated by $\frac{c}{\Omega} \sigma^{\frac{1}{2n-1}}$, the solution of the $l_0$-minimization problem in the single measurement case
\begin{equation}\label{prob:singlel0minimization}
\min_{\widehat \mu} \bnorm{\widehat \mu}_{0} \quad \text{subject to} \quad \bnorm{\mathcal F[\widehat \mu] -\vect Y}_{\infty}< \sigma, 
\end{equation}	
is not stable. In particular, the recovered point sources by (\ref{prob:singlel0minimization}) may be located in an interval completely disjoint from that of the ground truth. 


Therefore, for the single snapshot case, when the point sources are separated by $O(\frac{\sigma^{\frac{1}{2n-1}}}{\Omega})$, the 
$l_0$-minimization may be unstable. However, for the multi-snapshot spectral estimation, when the point sources are separated by $O(\frac{(\frac{\sigma}{\sigma_{\infty, \min}\left(\vect A^{\top}\right)})^{\frac{1}{n}}}{\Omega})$, the $l_0$-minimization (\ref{prob:l0minimization2}) is still stable. 

For example, in the DOA estimation, suppose we have a high degree of incoherence in the amplitude vectors $a_{j,t}$'s (such as Gaussian random variables in (\ref{equ:gaussianillum1})), making $\frac{1}{\sigma_{\infty, \min}\left(\vect A^{\top}\right)}$ or $\left(\frac{1}{\sigma_{\infty, \min}\left(\vect A^{\top}\right)}\right)^{\frac{1}{n}}$ be of constant order, the resolution now is of order $O(\frac{\sigma^{\frac{1}{n}}}{\Omega})$. Compared with resolution in the single measurement case, say of order $O(\frac{\sigma^{\frac{1}{2n-1}}}{\Omega})$, this clearly shows a significant enhancement and illustrates the effect of multiple snapshots in improving the resolution.

\subsection{Lower bound for the resolution estimate}
Theorem \ref{thm:l0normrecovery0} states that when we have a high degree of incoherence for columns of $\vect A^{\top}$ so that $\sigma_{\infty, \min}\left(\vect A^{\top}\right)$ is of order one, the resolution of the sparse recovery (\ref{prob:l0minimization2}) should be less than $\frac{c}{\Omega}\sigma^{\frac{1}{n}}$ for some positive constant $c$. The following proposition, proven in Appendix \ref{section:proofofsupportlowerbound}, illustrates that this resolution order is the best achievable. 


\begin{prop}\label{prop:multisupportlowerboundthm1}
	Given $n \geqs 2$, and source amplitudes $a_{j,t}, 1\leqs j\leqs n, 1\leqs t\leqs T$, let $\tau$ be given by	\begin{equation}\label{equ:multisupportlowerboundsepadis2}
	\tau = \frac{0.044}{\Omega}\Big(\frac{\sigma}{\bnorm{\vect A}_{1}}\Big)^{\frac{1}{n}}.
	\end{equation}
 with $\frac{\sigma}{\bnorm{\vect A}_1}\leqs 1$. Then there exist $n$ points $y_1=-\tau, y_2=-2\tau,\ldots, y_n=-n\tau $ and $n$ points $\widehat y_1= 0, \widehat \tau,\cdots, \widehat y_n= (n-1)\tau$, such that there exist $\widehat a_{j,t}$'s so that 
\[
    \bnorm{\mathcal F \left[\sum_{j=1}^{n}\widehat a_{j,t} \delta_{\widehat y_j}\right]-\mathcal F\left[\sum_{j=1}^n a_{j,t}\delta_{y_j}\right]}_{\infty}< \sigma, \ t=1, \cdots,  T.
\] 
\end{prop} 

In Proposition \ref{prop:multisupportlowerboundthm1}, the recovered location and the ground truth are in two intervals that are completely disjoint from each other, rendering stable location recovery impossible. In particular, normalizing the magnitude of the amplitude vector (or signal) such that $\sum_{j=1}^{n}\babs{a_{j,t}}=O(1), t=1,\cdots, T$, (\ref{equ:multisupportlowerboundsepadis2}) simplifies to 
\[
\tau \approx \frac{0.044}{\Omega}\sigma^{\frac{1}{n}},
\]
which corroborates the sharpness of estimate (\ref{equ:sepaconditionl0normrecovery}) for highly uncorrelated amplitude vectors.

We also conjecture that, for a given amplitude matrix $\vect A$, the optimal resolution lower-bound estimate should be of order $O\left(\frac{1}{\Omega}\left(\frac{\sigma}{\sigma_{\infty, \min}\left(\vect A^{\top}\right)}\right)^{\frac{1}{n}}\right)$, validating the sharpness of (\ref{equ:sepaconditionl0normrecovery}) in general cases. We hope to prove it in future research.

\subsection{Discrete model and MMV problems}
In this subsection, we introduce the discrete version of the model (\ref{equ:multimodelsetting1}) and the problem (\ref{prob:l0minimization2}). Specifically, we assume the $n$ sources $\delta_{y_j}$'s are supported on $N$ evenly spaced grid points $\{x_1, \cdots, x_N\}$ in $[-d, d]$. The collection of point sources in each snapshot is still
\begin{equation}\label{equ:sources2}
\mu_t = \sum_{j=1}^n a_{j, t} \delta_{y_j}, \quad t=1,\cdots, T,
\end{equation}
where $T$ is the total number of snapshots. We sample the Fourier transform of $\mu_t$ at $M>n$ equispaced points:
\begin{equation}\label{equ:multidiscretemodel1}
\mathbf{Y}_t\left(\omega_q\right)=\mathcal{F} [\mu_t]\left(\omega_q\right)+\mathbf{W}_t\left(\omega_q\right)=\sum_{j=1}^n a_{j,t} e^{i y_j \omega_q}+\mathbf{W}_t\left(\omega_q\right), \quad 1 \leqslant q \leqslant M, t=1,\cdots, T,
\end{equation}
where $\omega_1=-\Omega, \omega_2=-\Omega+h, \cdots, \omega_M=\Omega$. Here $h=\frac{2 \Omega}{M-1}$ is the sampling spacing and $\mathbf{W}_t\left(\omega_q\right)$ 's are the noise.  Throughout, we assume that $M>n$ and that $h \leqs \frac{\pi}{2d}$. The latter assumption excludes the non-uniqueness of the point source due to shifts by multiples of $\frac{2 \pi}{h}$. Denote
$$
\mathbf{Y}_t=\left(\mathbf{Y}_t\left(\omega_1\right), \cdots, \mathbf{Y}_t\left(\omega_M\right)\right)^{\top}, \text { and } \mathbf{W}_t=\left(\mathbf{W}_t\left(\omega_1\right), \cdots, \mathbf{W}\left(\omega_M\right)\right)^{\top}.
$$
Then (\ref{equ:multidiscretemodel1}) can be rewritten as 
\begin{equation}\label{equ:mmvproblem1}
\vect S = \vect \Phi \vect A + \vect W,
\end{equation}
where $\vect{S}= \left(\vect Y_1, \cdots, \vect Y_T\right)$, $\vect{W}= \left(\vect W_1, \cdots, \vect W_T\right)$, $\vect A$ is the amplitude matrix defined by (\ref{equ:illuminationpattern1}), and \begin{equation}\label{equ:defineofphi1}
\vect{\Phi} := \left(\psi_{M-1}(e^{i x_1}),\ \cdots,\ \psi_{M-1}(e^{i x_N})\right)
\end{equation}
with $\psi_{s}(z) := z^{-\Omega}(1, z,\cdots, z^s)^{\top}$ and $x_j$'s being the grid points. We assume that the noise matrix $\vect W$ satisfies 
\[
\bnorm{\vect W}_{\mathrm{F}}< \sigma 
\]
with $\sigma$ being the noise level. 

The problem (\ref{equ:mmvproblem1}) is known as the multiple measurement vectors (MMV) problem \cite{cotter2005sparse, chen2006theoretical, davies2012rank}, which is common in DOA estimation, machine learning, and compressive sensing. The sparse recovery associated with MMV \cite{davies2012rank, tropp2006algorithmsi, tropp2006algorithmsii} reads
\begin{equation}\label{equ:jointsparsitydiscretemodel2}
\min_{\mathbf{\widehat A} \in \mathbb{C}^{N \times T}}\bnorm{\mathbf{\widehat A}}_{\text {row}-0}\quad  \text{ subject to }\bnorm{\vect{S}-\vect{\Phi} \vect{\widehat A}}_{\mathrm{F}} <\sigma, 
\end{equation}
where the matrix $\mathbf{\widehat A} = \left(\widehat a_{j,t}\right)_{j=1,\cdots, N, t=1,\cdots, T}$ and $\bnorm{\mathbf{\widehat A}}_{\text{row}-0}$ is the row-$l_0$ quasi-norm \cite{tropp2006algorithmsi, tropp2006algorithmsii} defined as  
\[
\bnorm{\mathbf{\widehat A} }_{\text{row}-0}:=\# \operatorname{rowsupp}(\mathbf{\widehat A} ) 
\]
with $\operatorname{rowsupp}\left(\mathbf{\widehat A} \right):=\left\{j \in \{1,\cdots, N\}: \widehat a_{j,t} \neq 0 \text{ for some } t \in \{1,\cdots, T\} \right\}$. Note that this is also a transformation of the optimization (\ref{prob:l0minimization2}) in the discrete setting. The convex relaxation of this combinatorial problem is
$$
\min_{\mathbf{\widehat A} \in \mathbb{C}^{N \times T}}\bnorm{\mathbf{\widehat A}}_{\mathrm{rx}} \quad \text { subject to }\bnorm{\vect{S}-\vect{\Phi} \mathbf{\widehat A}}_{\mathrm{F}} < \sigma, 
$$
where $\bnorm{\mathbf{\widehat A}}_{\mathrm{rx}} := \sum_{1\leqslant j\leqslant N} \max_t\left|\widehat a_{j,t}\right|$. See \cite{tropp2006algorithmsi, tropp2006algorithmsii} for a detailed discussion of these problems and the proposed algorithms. In particular, \cite{tropp2006algorithmsi, tropp2006algorithmsii} focus on the simultaneous sparse approximation of signals consisting of general elementary atoms, and their stability findings (Theorem 5.1) do not apply to the super-resolution problem, due to the ill-condition of the corresponding dictionary.

Similarly to Theorem \ref{thm:l0normrecovery0}, we derive the following stability result for the optimization problem (\ref{equ:jointsparsitydiscretemodel2}), whose proof is given in Section \ref{section:proofdiscretel0normrecovery}. To the best of our knowledge, this is the first stability result for the sparse MMV recovery in the super-resolution problem.

\begin{thm}\label{thm:discretel0normrecovery0}
Consider the grid $\{x_j\}_{j=1}^N$ of length $\frac{c_0n\pi}{2\Omega}$. Let $n\geqs 2$ and $\vect Y_t$'s in (\ref{equ:multidiscretemodel1}) be the measurements generated by $n$ point sources $\delta_{y_j}$'s supported on the grid. Suppose that the following separation condition holds:	\begin{equation}\label{equ:mmvsepaconditionl0normrecovery}
	d_{\min} := \min_{p\neq j}\babs{y_p-y_j}\geqs \frac{4c_0e\pi }{\Omega }\Big(\frac{\sigma}{\sigma_{\min}(\vect A^{\top})}\Big)^{\frac{1}{n}},
	\end{equation}
	where $\sigma_{\min}(\vect A^{\top})$ is the minimum singular value of $\vect A^{\top}$ satisfying $\frac{\sigma}{\sigma_{\min}(\vect A^{\top})}\leqs 1$. Then, any solution $\mathbf{\widehat A}$ to (\ref{equ:jointsparsitydiscretemodel2}) contains exactly $n$ nonzero rows. Moreover, let $\widehat y_1, \cdots, \widehat y_n$ be points in the grid $\{x_j\}_{j=1}^N$ corresponding to the nonzero rows of a solution of (\ref{equ:jointsparsitydiscretemodel2}). After reordering the $\widehat y_j$'s, we have 
	\begin{equation}
	\Big|\widehat y_j-y_j\Big|<\frac{d_{\min}}{2},
	\end{equation} 
	and 
	\begin{equation}
 \Big|\widehat y_j-y_j\Big|< \frac{C(n)}{\Omega}\mathrm{SRF}^{n-1}\frac{\sigma}{\sigma_{\infty, \min}\left(\vect A\right)}, \quad 1\leqs j\leqs n,
	\end{equation}
	where $\mathrm{SRF} = \frac{\pi}{\Omega d_{\min}}$ and $C(n)=\sqrt{\pi}n^{\frac{3}{2}}c_0^{n}e^{n}$.
\end{thm}

\section{Multi-dimensional case}\label{section:highdcase}
\subsection{Problem setting}
In this section, we generalize Theorem \ref{thm:l0normrecovery0} to the $d$-dimensional space $\mathbb R^d$. Let us first introduce the model setting. Similarly to the one-dimensional case, the collection of $n$ point sources in each snapshot reads
\[
\mu_t = \sum_{j=1}^n a_{j,t}\delta_{\vect y_j},\quad t=1, \cdots, T, 
\]
where $\vect y_j\in \mathbb R^d, 1\leqs j\leqs n$ are the source locations and $a_{j,t}$ represents the amplitude of source $\delta_{\vect y_j}$ in the $t$-th snapshot. The available measurements are given by  
\begin{equation}\label{equ:highdmultimodelsetting1}
\mathbf Y_t(\vect \omega) = \mathcal F \left[\mu_t\right] (\vect \omega) + \mathbf W_t(\vect \omega)= \sum_{j=1}^{n}a_{j,t} e^{i \vect y_j \cdot \vect \omega} + \mathbf W_t(\vect \omega), \ 1\leqs t\leqs T,\ \bnorm{\vect \omega}_2 \leqs \Omega,
\end{equation}
where $\mathcal F\left[\mu_t\right]$ denotes the $d$-dimensional Fourier transform of $\mu_t$ and $\vect W_t(\vect \omega)$ is the noise. We assume that $\bnorm{\mathbf W_t}_{\infty}<\sigma$ with $\sigma$ being the noise level and $\bnorm{f}_\infty =\max_{\bnorm{\vect \omega}_2 \leqs \Omega }\babs{f(\vect \omega)}$. 

We consider reconstructing the point sources as the sparsest solution (solution to the $l_0$-minimization problem) under the measurement constraint. We suppose that the point sources are located in a sphere $\mathcal O\subset \mathbb R^d$ with a radius of several Rayleigh resolution limits. Then we consider the following optimization problem: 
\begin{equation}\label{prob:highdl0minimization2}
\begin{aligned}
&\min_{\{\mathbf{\widehat y}_1, \cdots, \mathbf{\widehat y}_k\}\subset \mathcal O} \left(\#\{\mathbf{\widehat y}_1, \cdots, \mathbf{\widehat y}_k\}\right) \\
&\text{subject to the existence of $\widehat a_{j,t}$'s such that}\  \bnorm{\mathcal F\left[\sum_{j=1}^n\widehat a_{j,t} \delta_{\mathbf{\widehat y}_j}\right] - \vect Y_t}_{\infty}< \sigma,\ 1\leqslant t\leqslant T.
\end{aligned}
\end{equation} 
Our main result in the following subsection gives resolution and stability estimates for the problem (\ref{prob:highdl0minimization2}).

\subsection{Main results for the stability of sparse recoveries in multi-dimensions}
The amplitude matrix is still
\begin{equation}\label{equ:highdilluminationpattern1}
\vect A = \begin{pmatrix}
a_{1,1}&\cdots&a_{1, T}\\
\vdots&\vdots&\vdots\\
a_{n,1}&\cdots&a_{n,T}\\
\end{pmatrix}.
\end{equation}
Define
\begin{equation}\label{equ:gammaformula1}
\gamma(d)= \begin{cases}
\sum_{j=1}^d \frac{1}{j}, & d \geqslant 1, \\ 
0, & d=0.
\end{cases}
\end{equation}
 We have the following theorem. We refer to Section \ref{section:proofhighdl0normrecovery}  for its proof.

 \begin{thm}\label{thm:highdl0normrecovery0}
 	Let $n\geqs 2$ and let the sphere $\mathcal O\subset \mathbb R^d$ be of radius  $\frac{c_0n\pi}{4\Omega}$ with $c_0\geqs 1$. Let $\vect Y_t$'s be the measurements that are generated by $n$ point sources at $\{\vect y_1, \cdots, \vect y_n\}, \vect y_j \in \mathcal O$ in the $d$-dimensional space. Assume that
  \begin{equation}\label{equ:highdsupportlimithm0equ0}
 		d_{\min}:=\min_{p\neq j}\Big|\Big|\mathbf y_p-\mathbf y_j\Big|\Big|_2\geqs \frac{2.4c_0e\pi4^{d-1}\left(\frac{(n+2)(n+1)}{2}\right)^{\gamma(d-1)}}{\Omega}\left(\frac{\sigma}{\sigma_{\infty, \min}\left(\vect A^{\top}\right)}\right)^{\frac{1}{n}}. 
 	\end{equation}
 	Then any solution to (\ref{prob:highdl0minimization2}) contains exactly $n$ points. Moreover, for $\{\mathbf{\widehat y}_1, \cdots, \mathbf{\widehat y}_n\}$ being a corresponding solution, after reordering the $\widehat {\mathbf y}_j$'s, we have 
 	\begin{equation}
 		\btwonorm{\widehat {\mathbf y}_j- \vect y_j}<\frac{d_{\min}}{2},
 	\end{equation} 
 	and 
 	\begin{equation}
 		\btwonorm{\widehat {\mathbf y}_j-\vect y_j} < \frac{C(d, n)}{\Omega}\mathrm{SRF}^{n-1}\frac{\sigma}{\sigma_{\infty, \min}\left(\vect A\right)}, \quad 1\leqs j\leqs n,
 	\end{equation}
 	where $C(d, n)= \sqrt{\pi}n^2c_0^{n}e^{n}4^{(d-1)n}\left(\frac{(n+2)(n+1)}{2}\right)^{\gamma(d-1)n}$ and $\mathrm{SRF} = \frac{\pi}{\Omega d_{\min}}$ is the super-resolution factor.
 \end{thm}

Theorem \ref{thm:highdl0normrecovery0} is the $d$-dimensional counterpart to Theorem \ref{thm:l0normrecovery0}. It reveals the dependence of the resolution and stability of $d$-dimensional sparse recoveries on the cut-off frequency, the noise level, the sparsity of point sources, and the incoherence of amplitude vectors. Moreover, it suggests that if we can ensure a small enough level of noise and a large enough value for $\sigma_{\infty, \min}
\left(\vect A^{\top}\right)$, super-resolution from multiple snapshots in multi-dimensional spaces becomes feasible.



\begin{remark}
 Compared to the estimate (\ref{equ:sepaconditionl0normrecovery}) for the one-dimensional problem, the estimate (\ref{equ:highdsupportlimithm0equ0}) contains a constant factor that depends on $d$ and $n$. We conjecture that it can be improved to a constant depending only on $d$. 
\end{remark}

\section{Non-linear approximation theory in Vandermonde space}\label{section:approxinvandermonde}
In this section, we present the main technique that is used in the proofs of the main results of the paper, namely the approximation theory in Vandermonde space. This theory was first introduced in \cite{liu2021mathematicaloned, liu2021theorylse}.
 Instead of considering the non-linear approximation problem there, we consider different approximation problems, which are relevant to the stability analysis of (\ref{prob:l0minimization2}) and (\ref{equ:jointsparsitydiscretemodel2}). More specifically,  for  $s \in \mathbb{N}, s \geqs 1,$ and $z\in \mathbb C$, we define the complex Vandermonde-vector
\begin{equation}\label{equ:multiphiformula}
\phi_s(z)=(1,z,\cdots,z^s)^\top.
\end{equation}
We consider the following non-linear problems:  
\begin{equation}\label{equ:multinon-linearapproxproblem1}
\min_{\widehat \theta_j \in \mathbb R, j=1,\cdots,k}\max_{t=1, \cdots, T} \min_{\widehat a_{j,t}\in \mathbb{C}, j=1,\cdots,k}\Big|\Big|\sum_{j=1}^k \widehat a_{j,t}\phi_s(e^{i\widehat \theta_j})-v_t\Big|\Big|_2,
\end{equation}
and 
\begin{equation}\label{equ:multinon-linearapproxproblem2}
\min_{\widehat \theta_j \in \mathbb R, j=1,\cdots,k}\sum_{t=1}^{T} \min_{\widehat a_{j,t}\in \mathbb{C}, j=1,\cdots,k}\Big|\Big|\sum_{j=1}^k \widehat a_{j,t}\phi_s(e^{i\widehat \theta_j})-v_t\Big|\Big|_2^2,
\end{equation}
where $v_t=\sum_{j=1}^{k+1}a_{j,t}\phi_s(e^{i\theta_j})$ is given with $\theta_j$'s being real numbers. We shall derive lower bounds for the optimal value of the minimization problem for the case when $s\leq k$. The main results are presented in Section \ref{section:mainresultsapproxinvandermonde}. 

\subsection{Notation and Preliminaries}
We first introduce some notation and preliminaries. We denote for $k \in \mathbb{N}, k\geqs 1$, 
\begin{equation}\label{equ:multizetaxiformula1} 
\zeta(k)= \left\{
\begin{array}{cc}
(\frac{k-1}{2}!)^2,& \text{$k$ is odd,}\\
(\frac{k}{2})!(\frac{k-2}{2})!,& \text{$k$ is even,}
\end{array} 
\right. \ \xi(k)=\left\{
\begin{array}{cc}
1/ 2,  & k=1,\\
\frac{(\frac{k-1}{2})!(\frac{k-3}{2})!}{4},& \text{$k$ is odd,\,\,$ k\geqs 3$,}\\
\frac{(\frac{k-2}{2}!)^2}{4},& \text{$k$ is even}.
\end{array} 
\right.	
\end{equation}
We also define for  $p, q \in \mathbb{N}, p,q \geqs 1$, and $z_1, \cdots, z_p, \widehat z_1, \cdots, \widehat z_q \in \mathbb C$, the following vector in $\mathbb{R}^p$:
\begin{equation}\label{equ:multieta}
\eta_{p,q}(z_1,\cdots,z_{p}, \widehat z_1,\cdots,\widehat z_q)=\left(\begin{array}{c}
|(z_1-\widehat z_1)|\cdots|(z_1-\widehat z_q)|\\
|(z_2-\widehat z_1)|\cdots|(z_2-\widehat z_q)|\\
\vdots\\
|(z_{p}-\widehat z_1)|\cdots|(z_{p}-\widehat  z_q)|
\end{array}\right).
\end{equation}

We present two lemmas from \cite[Section III]{liu2021theorylse} that help derive our main results. 


\begin{lem}\label{lem:multimultiproductlowerbound1}
	For $\theta_j \in \left[-\frac{\pi}{2}, \frac{\pi}{2}\right], j=1, \cdots, k+1$, assume that $\min_{p\neq j}|\theta_p-\theta_j|=\theta_{\min}$. Then, for any $\widehat \theta_1,\cdots, \widehat \theta_k\in \mathbb R$, we have the following estimate: 	
	\[
	\bnorm{\eta_{k+1,k}(e^{i\theta_1},\cdots,e^{i\theta_{k+1}},e^{i \widehat \theta_1},\cdots,e^{i\widehat \theta_k})}_{\infty}\geqs \xi(k)\left(\frac{2 \theta_{\min}}{\pi}\right)^k.  
	\] 
\end{lem}

\medskip
\begin{lem}\label{lem:multistablemultiproduct0}
Let $\epsilon >0$. For $\theta_j, \widehat \theta_j \in \left[-\frac{\pi}{2}, \frac{\pi}{2}\right], j=1, \cdots, k$, assume that
\begin{equation}\label{equ:stablemultiproductlemma1equ1}
	\bnorm{\eta_{k,k}(e^{i\theta_1},\cdots,e^{i\theta_k}, e^{i\widehat \theta_1},\cdots, e^{i\widehat \theta_k})}_{\infty}< \left(\frac{2}{\pi}\right)^{k}\epsilon,
	\end{equation}
	where $\eta_{k,k}$ is defined as in (\ref{equ:multieta}), and that
	\begin{equation}\label{equ:stablemultiproductlemma1equ2}
	\theta_{\min} =\min_{q\neq j}|\theta_q-\theta_j|\geqs \Big(\frac{4\epsilon}{\lambda(k)}\Big)^{\frac{1}{k}},
	\end{equation} 
	where 
	\begin{equation}\label{equ:lambda1}
	\lambda(k)=\left\{
	\begin{array}{ll}
	1,  & k=2,\\
	\xi(k-2),& k\geqs 3.
	\end{array} 
	\right.	
	\end{equation}
	Then, after reordering the $\widehat \theta_j$'s, we have
	\begin{equation}\label{equ:stablemultiproductlemma1equ4}
	\babs{\widehat \theta_j -\theta_j}< \frac{\theta_{\min}}{2},  \quad j=1,\cdots,k,
	\end{equation}
	and moreover,
	\begin{equation}\label{equ:stablemultiproductlemma1equ5}
	\babs{\widehat \theta_j -\theta_j}< \frac{2^{k-1}\epsilon}{(k-2)!(\theta_{\min})^{k-1}}, \quad j=1,\cdots, k.
	\end{equation}
\end{lem}

\subsection{Approximation theory in Vandermonde space}\label{section:mainresultsapproxinvandermonde}
Before presenting a lower bound for problems (\ref{equ:multinon-linearapproxproblem1}) and (\ref{equ:multinon-linearapproxproblem2}), we introduce a basic approximation result in Vandermonde space. This result was first derived in \cite{liu2021theorylse}. 

\begin{thm}\label{thm:multispaceapprolowerbound0}
	Let $k\geqs 1$. For fixed $\widehat \theta_1,\cdots, \widehat \theta_k\in \mathbb{R}$, denote $\widehat A= \big(\phi_{k}(e^{i\widehat \theta_1}),\cdots, \phi_{k}(e^{i\widehat \theta_k})\big)$, where  the $\phi_{k}(e^{i\widehat \theta_j})$'s  are defined as in (\ref{equ:multiphiformula}). Let $V$ be the $k$-dimensional complex space spanned by the column vectors of $\widehat A$ and let $V^\perp$ be the one-dimensional orthogonal complement of $V$ in $\mathbb{C}^{k+1}$. Denote by $P_{V^{\perp}}$ the orthogonal projection onto $V^{\perp}$ in $\mathbb{C}^{k+1}$. Then, we have
	\begin{equation*}
	\min_{\widehat a\in \mathbb C^{k}}\bnorm{\widehat A\widehat a-\phi_{k}(e^{i\theta})}_2=\bnorm{P_{V^{\perp}}\big(\phi_{k}(e^{i\theta})\big)}_2 = \babs{v^*\phi_{k}(e^{i\theta})}\geqs \frac{1}{2^k}\babs{\Pi_{j=1}^k(e^{i\theta}-e^{i\widehat \theta_j})},
	\end{equation*}	
	where $v$ is a unit vector in $V^\perp$ and $v^*$ is its conjugate transpose. 	
\end{thm}

\medskip
We then have the following results for non-linear approximation (\ref{equ:multinon-linearapproxproblem1}) in Vandermonde space.
\begin{thm}\label{thm:multispaceapprolowerbound1}
	Let $k\geqs 1$ and $\theta_j\in\left[-\frac{\pi}{2}, \frac{\pi}{2} \right], 1\leqs j\leqs k+1,$ be $k+1$ distinct points with $\theta_{\min}=\min_{p\neq j}|\theta_p-\theta_j|>0$. For $q\leqs k$, let $\widehat \alpha_t(q)=(\widehat a_{1,t},\cdots, \widehat a_{q,t})^\top$, $\alpha_t=(a_{1,t},\cdots, a_{k+1, t})^\top$ and
	\[
	\widehat A(q)= \big(\phi_{k}(e^{i\widehat \theta_1}),\cdots, \phi_{k}(e^{i\widehat \theta_q})\big), \quad A= \big(\phi_{k}(e^{i\theta_1}),\cdots, \phi_{k}(e^{i\theta_{k+1}})\big),
	\]
	where $\phi_{k}(z)$ is defined as in (\ref{equ:multiphiformula}). Then, for any $\ \widehat \theta_1, \cdots,  \widehat \theta_q\in \mathbb{R}$,
	\begin{equation*}
	\max_{t=1,\cdots,T}\min_{\widehat \alpha_t(q)\in \mathbb C^q}\bnorm{\widehat A(q)\widehat \alpha_t(q)-A \alpha_t}_2\geqs  \frac{\sigma_{\infty, \min}(B)\xi(k)(\theta_{\min})^{k}}{\pi^{k}},
	\end{equation*}
	where 
	\begin{equation}\label{equ:multispaceapprolowerbound1equ1}
	B=\left(\begin{array}{cccc}
	a_{1,1}&a_{2,1}&\cdots&a_{k+1,1}\\
	\vdots&\vdots&\vdots&\vdots\\
	a_{1,T}&a_{2,T}&\cdots&a_{k+1, T}
	\end{array}\right).
	\end{equation}
\end{thm}
\begin{proof}
\textbf{Step 1.} 
Note that, for any  $\widehat \theta_1, \cdots, \widehat \theta_q, \cdots, \widehat \theta_k\in \mathbb R$, if $q<k$, then 
\[
\min_{\widehat a_t(q)\in \mathbb C^q}\bnorm{\widehat A(q)\widehat \alpha_t(q)-A\alpha_t}_2\geqs \min_{\widehat \alpha_t(k)\in \mathbb C^k}\bnorm{\widehat A(k)\widehat \alpha_t(k)-A\alpha_t}_2, \ 1\leqs t\leqs T.
\]
So, we only need to consider the case when $q=k$. We shall verify that, for any $k$ distinct points $\widehat \theta_1,\cdots,\widehat \theta_{k}\in \mathbb R$, we have 
\begin{equation}\label{equ:multispaceapproxlowerboundequ-3}
\max_{t=1,\cdots,T}\min_{\widehat \alpha_t(k)\in \mathbb{C}^k}\bnorm{\widehat A(k)\widehat \alpha_t(k)-A\alpha_t}_2\geqs \frac{\sigma_{\infty, \min}(B)\xi(k)(\theta_{\min})^{k}}{\pi^{k}}.
\end{equation}
Let us then fix $(\widehat \theta_1,\cdots,\widehat \theta_{k})$ in our subsequent arguments.\\

\textbf{Step 2.}  Let $V$ be the complex space spanned by the column vectors of $\widehat A(k)$ and let $V^\perp$ be the orthogonal complement of $V$ in $\mathbb C^{k+1}$. It is clear that $V^\perp$ is a one-dimensional complex space. We let $v$ be a unit vector in $V^{\perp}$ and denote by $P_{V^{\perp}}$ the orthogonal projection onto $V^{\perp}$ in $\mathbb C^{k+1}$. Note that $\|P_{V^{\perp}}u\|_{2} =  |v^* u| $ for $u \in \mathbb C^{k+1}$, where
$v^*$ is the conjugate transpose of $v$. We have
\begin{align}\label{equ:multispaceapproxlowerboundequ-2}
\min_{\widehat \alpha_t\in \mathbb C^k}\bnorm{\widehat A(k)\widehat \alpha_t-A\alpha_t}_2 = \bnorm{P_{V^{\perp}}(A\alpha_{t})}_2= \babs{v^*A\alpha_{t}}= \babs{\sum_{j=1}^{k+1}a_{j,t}v^*\phi_{k}(e^{i\theta_j})} =\babs{\beta_{t}},
\end{align} 
where $\beta_{t}= \sum_{j=1}^{k+1} a_{j,t}v^*\phi_{k}(e^{i\theta_j}), \ t = 1,\cdots, T$. Denote by $\beta=(\beta_{1}, \beta_{2},\cdots, \beta_{T})^\top$. Thus, we only need to estimate the lower bound of $||\beta||_{\infty}$. By (\ref{equ:multispaceapproxlowerboundequ-2}), we have $\beta= B \widehat \eta$,
where $B$ is given by (\ref{equ:multispaceapprolowerbound1equ1}) and 
\begin{equation}\label{equ:multispaceapproxlowerboundequ-5}
\widehat \eta = (v^*\phi_{k}(e^{i\theta_1}), v^*\phi_{k}(e^{i\theta_2}), \cdots, v^*\phi_{k}(e^{i\theta_{k+1}}))^\top.
\end{equation}
By the definition of $\sigma_{\infty, \min}(B)$, we have 
\begin{equation}\label{equ:multispaceapproxlowerboundequ-4}
\bnorm{\beta}_{\infty}\geqs \sigma_{\infty, \min}(B)\bnorm{\widehat \eta}_{\infty}.
\end{equation}
On the other hand, by Theorem \ref{thm:multispaceapprolowerbound0}, we obtain that
\begin{equation}\label{equ:multispaceapproxlowerbound1equ2}
\bnorm{\widehat \eta}_{\infty} \geqs \frac{1}{2^k}\bnorm{\eta_{k+1, k}(e^{i\theta_1}, \cdots, e^{i\theta_{k+1}},  e^{i\widehat \theta_1}, \cdots, e^{i\widehat \theta_k})}_{\infty},
\end{equation}
where $\eta_{k+1,k}$ is defined by (\ref{equ:multieta}). Combining this with Lemma \ref{lem:multimultiproductlowerbound1}, we get 
$$
\bnorm{\widehat \eta}_{\infty} \geqs \frac{1}{2^k}\xi(k)\left(\frac{2 \theta_{\min}}{\pi}\right)^k.
$$ 
It then follows that 
\begin{align*}
\bnorm{\beta}_{\infty}\geqs \frac{\sigma_{\infty, \min}(B)\xi(k)(\theta_{\min})^{k}}{\pi^{k}},
\end{align*}
which proves (\ref{equ:multispaceapproxlowerboundequ-3}) and hence the theorem.  \end{proof}

\medskip
\begin{thm}\label{thm:multispaceapproxlowerbound2}
	Let $k\geqs 2$ and $\theta_j \in \mathbb R, j=1, \cdots, k$ be $k$ points. Assume that there are $k$ distinct points $\widehat \theta_1,\cdots,\widehat \theta_k\in\left[-\frac{\pi}{2}, \frac{\pi}{2} \right]$ satisfying
	\[ 
 \max_{t=1, \cdots, T} ||\widehat A\widehat \alpha_t-A \alpha_t||_2< \sigma, \]
	where
	$\widehat \alpha_t=(\widehat a_{1,t},\cdots, \widehat a_{k,t})^\top$, $\alpha_t=(a_{1,t},\cdots, a_{k,t})^\top$ and
	\[
	\widehat A= \big(\phi_{k}(e^{i \widehat \theta_1}),\cdots, \phi_{k}(e^{i \widehat \theta_k})\big), \quad A= \big(\phi_{k}(e^{i \theta_1}),\cdots, \phi_{k}(e^{i \theta_{k}})\big).
	\]
	Then
	\[
	\bnorm{\eta_{k,k}(e^{i \theta_1},\cdots,e^{i \theta_k},e^{i \widehat \theta_1},\cdots,e^{i \widehat \theta_k})}_{\infty}<\frac{2^{k}}{\sigma_{\infty, \min}(B)}\sigma,
	\]
	where
	\begin{equation}\label{equ:multispaceapproxlowerbound2equ1}
	B=\left(\begin{array}{cccc}
	a_{1,1}&a_{2,1}&\cdots&a_{k,1}\\
	\vdots&\vdots&\vdots&\vdots\\
	a_{1,T}&a_{2,T}&\cdots&a_{k, T}
	\end{array}\right).
	\end{equation}
\end{thm}
\begin{proof} Let $V$ be the complex space spanned by the column vectors of $\widehat A$ and let  $V^\perp$ be the orthogonal complement of $V$ in $\mathbb C^{k+1}$.  Let $v$ be a unit vector in $V^{\perp}$ and denote by $P_{V^{\perp}}$ the orthogonal projection onto $V^{\perp}$ in $\mathbb C^{k+1}$. Similarly to {Step 2} in the proof of Theorem \ref{thm:multispaceapprolowerbound1}, we obtain that
\begin{align}\label{equ:multispaceapproxlowerbound2equ2}
\min_{\widehat \alpha_t\in \mathbb C^k}||\widehat A\widehat \alpha_t-A\alpha_t||_2 = ||P_{V^{\perp}}(A\alpha_{t})||_2= |v^*A\alpha_{t}|= |\sum_{j=1}^{k}a_{j,t}v^*\phi_{k}(e^{i\theta_j})| =|\beta_{t}|,
\end{align} 
where $\beta_{t}= \sum_{j=1}^{k} a_{j,t}v^*\phi_{k}(e^{i\theta_j}), \ t = 1,\cdots, T$. Denote by $\beta=(\beta_{1}, \beta_{2},\cdots, \beta_{T})^\top$, we have $\beta= B \widehat \eta$,
where $B$ is given by (\ref{equ:multispaceapproxlowerbound2equ1}) and 
\begin{equation}\label{equ:multispaceapproxlowerbound2equ3}
\widehat \eta = (v^*\phi_{k}(e^{i\theta_1}), v^*\phi_{k}(e^{i\theta_2}), \cdots, v^*\phi_{k}(e^{i\theta_{k}}))^\top. 
\end{equation}
By the definition of $\sigma_{\infty, \min}(B)$, we arrive at
\[
||\beta||_{\infty}\geqs \sigma_{\infty, \min}(B)||\widehat \eta||_{\infty}.
\]
On the other hand, by Theorem \ref{thm:multispaceapprolowerbound0}, we get
\begin{equation*}
\bnorm{\widehat \eta}_{\infty} \geqs \frac{1}{2^k}\bnorm{\eta_{k, k}(e^{i\theta_1}, \cdots, e^{i\theta_{k}},  e^{i\widehat \theta_1}, \cdots, e^{i\widehat \theta_k})}_{\infty},
\end{equation*}
and hence the theorem is proved. 
\end{proof}

\medskip

Lastly, we present stability results for the non-linear approximation (\ref{equ:multinon-linearapproxproblem2}) that are similar to Theorems \ref{thm:multispaceapprolowerbound1} and \ref{thm:multispaceapproxlowerbound2}.
\begin{thm}\label{thm:multispaceapprolowerbound3}
	Let $k\geqs 1$ and $\theta_j\in\left[-\frac{\pi}{2}, \frac{\pi}{2} \right], 1\leqs j\leqs k+1,$ be $k+1$ distinct points with $\theta_{\min}=\min_{p\neq j}|\theta_p-\theta_j|>0$. For $q\leqs k$, let $\widehat \alpha_t(q)=(\widehat a_{1,t},\cdots, \widehat a_{q,t})^\top$, $\alpha_t=(a_{1,t},\cdots, a_{k+1, t})^\top$ and
	\[
	\widehat A(q)= \big(\phi_{k}(e^{i\widehat \theta_1}),\cdots, \phi_{k}(e^{i\widehat \theta_q})\big), \quad A= \big(\phi_{k}(e^{i\theta_1}),\cdots, \phi_{k}(e^{i\theta_{k+1}})\big),
	\]
	where $\phi_{k}(z)$ is defined as in (\ref{equ:multiphiformula}). Then, for any $\ \widehat \theta_1, \cdots,  \widehat \theta_q\in \mathbb{R}$,
	\begin{equation*}
	\sum_{t=1}^{T}\min_{\widehat \alpha_t(q)\in \mathbb C^q}\bnorm{\widehat A(q)\widehat \alpha_t(q)-A \alpha_t}_2^2\geqs  \frac{\sigma_{\min}(B)^2\xi(k)^2(\theta_{\min})^{2k}}{\pi^{2k}},
	\end{equation*}
	where $\sigma_{\min}(B)$ is the minimum singular value of $B$ in (\ref{equ:multispaceapprolowerbound1equ1}).
\end{thm}
\begin{proof}
Like Step 1 in the proof of Theorem \ref{thm:multispaceapprolowerbound1}, we only need to prove that, for any $k$ distinct points $\widehat \theta_1,\cdots,\widehat \theta_{k}\in \mathbb R$, we have 
\begin{equation}\label{equ:proofdiscreteapprox1}
\sum_{t=1}^T\min_{\widehat \alpha_t(k)\in \mathbb{C}^k}\bnorm{\widehat A(k)\widehat \alpha_t(k)-A\alpha_t}_2^2\geqs \frac{\sigma_{\min}(B)^2\xi(k)^2(\theta_{\min})^{2k}}{\pi^{2k}}.
\end{equation}
 Denote by $\beta = (\beta_1, \beta_2, \cdots, \beta_T)^{\top}$ for $\beta_t$'s in (\ref{equ:multispaceapproxlowerboundequ-2}). We only need to estimate the lower bound of $\bnorm{\beta}_{2}^2$. By (\ref{equ:multispaceapproxlowerboundequ-2}), we have $\beta= B \widehat \eta$,
where $B, \widehat \eta$ are given by (\ref{equ:multispaceapprolowerbound1equ1}), (\ref{equ:multispaceapproxlowerboundequ-5}), respectively. Thus, we have the following estimate:
\begin{equation*}
\bnorm{\beta}_{2}\geqs \sigma_{\min}(B)\bnorm{\widehat \eta}_{2}.
\end{equation*}
On the other hand, by Theorem \ref{thm:multispaceapprolowerbound0}, we obtain that
\begin{equation*}
\bnorm{\widehat \eta}_{2} \geqs \bnorm{\widehat \eta}_{\infty} \geqs \frac{1}{2^k}\bnorm{\eta_{k+1, k}(e^{i\theta_1}, \cdots, e^{i\theta_{k+1}},  e^{i\widehat \theta_1}, \cdots, e^{i\widehat \theta_k})}_{\infty},
\end{equation*}
where $\eta_{k+1,k}$ is defined by (\ref{equ:multieta}). Combining this with Lemma \ref{lem:multimultiproductlowerbound1}, we get 
$$
\bnorm{\widehat \eta}_{2} \geqs \frac{1}{2^k}\xi(k)\left(\frac{2 \theta_{\min}}{\pi}\right)^k.
$$ 
It then follows that 
\begin{align*}
\bnorm{\beta}_{2}^2\geqs \frac{\sigma_{\min}(B)^2\xi(k)^2(\theta_{\min})^{2k}}{\pi^{2k}},
\end{align*}
which proves (\ref{equ:proofdiscreteapprox1}) and hence the theorem. \end{proof}

\begin{thm}\label{thm:multispaceapproxlowerbound4}
	Let $k\geqs 2$ and $\theta_j \in \mathbb R, j=1, \cdots, k$ be $k$ points. Assume that there are $k$ distinct points $\widehat \theta_1,\cdots,\widehat \theta_k\in\left[-\frac{\pi}{2}, \frac{\pi}{2} \right]$ satisfying
	\begin{equation}\label{equ:discreteapprox2}
 \sum_{t=1}^{T} ||\widehat A\widehat \alpha_t-A \alpha_t||_2^2 < \sigma^2, 
 \end{equation}
	where
	$\widehat \alpha_t=(\widehat a_{1,t},\cdots, \widehat a_{k,t})^\top$, $\alpha_t=(a_{1,t},\cdots, a_{k,t})^\top$ and
	\[
	\widehat A= \big(\phi_{k}(e^{i \widehat \theta_1}),\cdots, \phi_{k}(e^{i \widehat \theta_k})\big), \quad A= \big(\phi_{k}(e^{i \theta_1}),\cdots, \phi_{k}(e^{i \theta_{k}})\big).
	\]
	Then,
	\[
	\bnorm{\eta_{k,k}(e^{i \theta_1},\cdots,e^{i \theta_k},e^{i \widehat \theta_1},\cdots,e^{i \widehat \theta_k})}_{\infty}<\frac{2^{k}}{\sigma_{\min}(B)}\sigma,
	\]
	where $\sigma_{\min}(B)$ is the minimum singular value of $B$ in (\ref{equ:multispaceapproxlowerbound2equ1}).
\end{thm}
\begin{proof}
Denote by $\beta=(\beta_{1}, \beta_{2},\cdots, \beta_{T})^\top$ for $\beta_t$'s in (\ref{equ:multispaceapproxlowerbound2equ2}).  Then, we have $\beta= B \widehat \eta$,
where $B$ is given by (\ref{equ:multispaceapproxlowerbound2equ1}) and 
$\widehat \eta$ is given by (\ref{equ:multispaceapproxlowerbound2equ3}). By the definition of $\sigma_{\min}(B)$, we arrive at
\[
||\beta||_{2}\geqs \sigma_{\min}(B)||\widehat \eta||_{2}.
\]
By (\ref{equ:discreteapprox2}) and the definition of $\beta$, we further have 
\[
\sigma_{\min}(B)^2||\widehat \eta||_{2}^2\leqs ||\beta||_{2}^2<\sigma^2.
\]
On the other hand, by Theorem \ref{thm:multispaceapprolowerbound0}, we get
\begin{equation*}
\bnorm{\widehat \eta}_{\infty} \geqs \frac{1}{2^k}\bnorm{\eta_{k, k}(e^{i\theta_1}, \cdots, e^{i\theta_{k}},  e^{i\widehat \theta_1}, \cdots, e^{i\widehat \theta_k})}_{\infty},
\end{equation*}
and hence the theorem is proved. 
\end{proof}

\section{Proofs of Theorems \ref{thm:l0normrecovery0} and \ref{thm:discretel0normrecovery0}} \label{section:proofofthml0normrecover}

\subsection{Proof of Theorem \ref{thm:l0normrecovery0}}
We first introduce the following Theorems \ref{thm:onednumberbound}  and \ref{thm:onedsupportbound} for the number and location recoveries in the one-dimensional super-resolution problem. Leveraging the two theorems, we prove Theorem \ref{thm:l0normrecovery0} at the end of this section.

\begin{thm}\label{thm:onednumberbound}
Suppose that the measurements $\vect Y_t$'s in (\ref{equ:multimodelsetting1}) are generated from $n$ point sources located at $y_j$'s that are in an interval $\mathcal O$ of length $\frac{c_0n\pi}{2\Omega}$ with $c_0\geqs 1$ and satisfy
\begin{equation}\label{equ:onedsepaconditionnumber}
d_{\min} := \min_{p\neq j}\babs{y_p-y_j}\geqs \frac{2.4c_0e\pi }{\Omega }\Big(\frac{\sigma}{\sigma_{\infty, \min}\left(\vect A^{\top}\right)}\Big)^{\frac{1}{n}}
\end{equation}
with  $\frac{\sigma}{\sigma_{\infty, \min}\left(\vect A^{\top}\right)}\leqs 1$ and $\vect A$ being the amplitude matrix (\ref{equ:illuminationpattern1}). Then there is no $k<n$ locations $\widehat y_j\in \mathbb R, j=1, \cdots,k,$ such that there exists $\widehat \mu_t = \sum_{j=1}^k \widehat a_{j,t} \delta_{\widehat y_j}$'s so that 
\[
\bnorm{\mathcal F[\widehat \mu_t] - \vect Y_t}_{\infty}< \sigma, \quad t=1, \cdots, T. 
\] 
\end{thm}
\begin{proof} 
\textbf{Step 1.} Without loss of generality, we assume $\mathcal O = \left[-\frac{c_0 n\pi}{4\Omega}, \frac{c_0n\pi}{4\Omega} \right]$. For location set $\{y_1, \cdots, y_n\}$ of the underlying sources and a location set $\{\widehat y_1, \cdots, \widehat y_k\}$, we write $\widehat \mu_t = \sum_{j=1}^k\widehat a_{j,t} \delta_{\widehat y_j}$ and $\mu_{t} = \sum_{j=1}^n a_{j,t} \delta_{y_{j}}$. We shall prove that if $k<n$, then for any $\widehat y_j\in \mathbb R, \widehat a_{j,t} \in \mathbb C, j=1,\cdots,k, t=1,\cdots,T$, 
\begin{equation}\label{equ:multinumberresultequ0}
\max_{t=1, \cdots, T}\bnorm{\mathcal F\left[\widehat \mu_t\right]-\mathcal F\left[\mu_t\right]}_{\infty}> 2\sigma.
\end{equation}
In view of (\ref{equ:noiseconstraint1}), i.e. $\bnorm{\vect W_t}_{\infty}<\sigma$, from (\ref{equ:multinumberresultequ0}) we further have 
\begin{equation*}
\max_{t=1, \cdots, T} \bnorm{\mathcal F\left[\widehat \mu_t\right]-\vect Y_t}_{\infty}>\sigma,
\end{equation*} 
whereby any location set consisting of only $k<n$ elements cannot be a solution to (\ref{prob:l0minimization2}).

Specifically, for $k<n$, we consider 
\begin{equation}\label{equ:highdmultinumberresultequ1}
\left(\mathcal F [\widehat \mu_t](\omega_1), \mathcal F [\widehat \mu_t](\omega_2), \cdots,\mathcal F [\widehat \mu_t](\omega_{n})\right)^\top -\left(\mathcal F [\mu_t](\omega_1), \mathcal F  [\mu_t](\omega_2), \cdots,\mathcal F [\mu_t](\omega_{n})\right)^\top,
\end{equation}
where $\omega_j=(j-1)\tilde h-\Omega, \ j= 1, \cdots, n$ with $\tilde h = \frac{2\Omega}{c_0n}$. We write (\ref{equ:highdmultinumberresultequ1}) as 
\begin{equation*}
\widehat \Phi \widehat \alpha_t- \Phi \alpha_t, 
\end{equation*}
where $\widehat \alpha_t= \left(\widehat a_{1,t},\cdots, \widehat a_{k,t}\right)^\top$, $\alpha_t=\left(a_{1,t}, \cdots, a_{n,t}\right)^\top$ and  
\begin{equation*}
\widehat \Phi= \left(
\begin{array}{ccc}
e^{i\widehat y_1\omega_1}&\cdots& e^{i\widehat y_k\omega_1}\\
e^{i\widehat y_1\omega_{2}} &\cdots& e^{i\widehat y_k\omega_{2}}\\
\vdots&\vdots&\vdots\\
e^{i\widehat y_1\omega_{n}}&\cdots& e^{i\widehat y_k \omega_{n}}\\
\end{array}
\right), \quad \Phi=\left(
\begin{array}{ccc}
e^{i y_1\omega_1}&\cdots& e^{i y_n\omega_1}\\
e^{i y_1\omega_{2}} &\cdots& e^{i y_n\omega_{2}}\\
\vdots&\vdots&\vdots\\
e^{i y_1\omega_{n}}&\cdots& e^{i y_n \omega_{n}}
\end{array}
\right).
\end{equation*}
We shall prove that the following estimate holds:
\begin{equation}\label{equ:onedmultinumberresultequ3}
\max_{t=1, \cdots, T} \frac{1}{\sqrt{n}}\bnorm{\widehat \Phi \widehat \alpha_t- \Phi \alpha_t}_2> 2\sigma,
\end{equation}
and consequently it yields (\ref{equ:multinumberresultequ0}). 

\textbf{Step 2.} 
Let $\theta_j = y_j \tilde h =  y_j\frac{2\Omega}{c_0n}$ and $\widehat \theta_j = \widehat y_j \tilde h =\widehat y_j\frac{2\Omega}{c_0n}$. From the following decompositions:  
\begin{equation}\label{equ:multimatrixdecomposition1}
\begin{aligned}
&\widehat \Phi=\big(\phi_{n-1}(e^{i \widehat \theta_1}), \cdots,\phi_{n-1}(e^{i\widehat \theta_k} ) \big)\text{diag}(e^{-i\widehat y_1\Omega},\cdots,e^{-i\widehat y_k\Omega}),\\
&\Phi=\big(\phi_{n-1}(e^{i \theta_1}), \cdots,\phi_{n-1}(e^{i \theta_n})\big)\text{diag}(e^{-i y_1\Omega},\cdots,e^{-iy_n\Omega}),
\end{aligned}
\end{equation}
where $\phi_{n-1}(\cdot)$ is defined as in (\ref{equ:multiphiformula}), we readily obtain that
\begin{align} \label{equ:multi1111}
\max_{t=1, \cdots, T} \bnorm{\widehat \Phi \widehat \alpha_t-\Phi \alpha_t}_2= \max_{t=1, \cdots, T} \bnorm{\widehat D \widehat \gamma_t-D \gamma_t}_2,
\end{align}
where $\widehat \gamma_{t} = \left(\widehat a_{1,t}e^{-i\widehat y_1\Omega},\cdots, \widehat a_{k,t}e^{-i\widehat y_k\Omega}\right)^\top, \gamma_t=\left(a_{1,t}e^{-iy_1\Omega},\cdots, a_{n,t}e^{-iy_n\Omega}\right)^\top$, \\
$\widehat D=\big(\phi_{n-1}(e^{i \widehat \theta_1}), \cdots,\phi_{n-1}(e^{i \widehat \theta_k} ) \big)$ and $D=\big(\phi_{n-1}(e^{i \theta_1}), \cdots,\phi_{n-1}(e^{i \theta_n})\big)$.
We consider $\vect A$ in (\ref{equ:illuminationpattern1}) and denote $B= \vect A^{\top} \text{diag}\left(e^{-iy_1\Omega},\cdots, e^{-iy_n\Omega}\right)$. Since $\mathcal O = \left[-\frac{c_0 n\pi}{4\Omega}, \frac{c_0n\pi}{4} \right]$, we have $\theta_j \in \left[ -\frac{\pi}{2}, \frac{\pi}{2}\right], j=1,\cdots, n$. Applying Theorem \ref{thm:multispaceapprolowerbound1}, we get 
\begin{equation*}
\max_{t=1, \cdots, T} \bnorm{\widehat D \widehat \gamma_t-D \gamma_t}_2\geqs   \frac{\sigma_{\infty, \min}(B)\xi(n-1)(\theta_{\min})^{n-1}}{\pi^{n-1}}, 
\end{equation*}
where $\theta_{\min}=\min_{j\neq p}|\theta_j-\theta_p|$.
On the other hand, by the definition of $\sigma_{\infty, \min}(\cdot)$, we have
\[
\sigma_{\infty, \min}(B) = \sigma_{\infty, \min}\left(\vect A^{\top}\right).
\]
Thus,
\begin{equation*}
\max_{t=1, \cdots, T}\bnorm{\widehat D \widehat \gamma_t-D \gamma_t}_2\geqs   \frac{\sigma_{\infty, \min}(\vect A^{\top})\xi(n-1)(\theta_{\min})^{n-1}}{\pi^{n-1}}. 
\end{equation*}
By (\ref{equ:multi1111}), it follows that 
\[
\max_{t=1, \cdots, T} \bnorm{\widehat \Phi \widehat \alpha_t-\Phi \alpha_t}_2 \geqs   \frac{\sigma_{\infty, \min}\left(\vect A^{\top}\right)\xi(n-1)(\theta_{\min})^{n-1}}{\pi^{n-1}}. 
\]
By $\theta_j = y_j \tilde h$ we have $\theta_{\min} = \min_{p\neq j}|\theta_j-\theta_p|=d_{\min}\frac{2\Omega}{c_0n}$. Then the separation condition (\ref{equ:onedsepaconditionnumber}) and $\frac{\sigma}{\sigma_{\infty, \min}\left(\vect A^{\top}\right)} \leqs 1$ imply that
\begin{equation*}
\theta_{\min}\geqs \frac{4.8 e\pi }{n}\Big(\frac{\sigma}{\sigma_{\infty, \min}\left(\vect A^{\top}\right)}\Big)^{\frac{1}{n}} \geqs \frac{4.8 e\pi }{n}\Big(\frac{\sigma}{\sigma_{\infty, \min}\left(\vect A^{\top}\right)}\Big)^{\frac{1}{n-1}}> \pi \Big(\frac{2\sqrt{n}\sigma}{\xi(n-1)\sigma_{\infty, \min}\left(\vect A^{\top}\right)}\Big)^{\frac{1}{n-1}},
\end{equation*}
where we have used Lemma \ref{lem:multinumbercalculate1} for deriving the last inequality. Therefore, 
\[
\max_{t=1, \cdots, T} \bnorm{\widehat \Phi \widehat \alpha_t-\Phi \alpha_t}_2 > 2\sqrt{n}\sigma,
\]
whence we prove (\ref{equ:onedmultinumberresultequ3}).
\end{proof}

\begin{thm}\label{thm:onedsupportbound}
Suppose that the measurements $\vect Y_t$'s in (\ref{equ:multimodelsetting1}) are generated from $n$ point sources located at $y_j$'s that are in an interval $\mathcal O$ of length $\frac{c_0n\pi}{2\Omega}$ with $c_0\geqs 1$ and satisfy
\begin{equation}\label{equ:onedsepaconditionsupport}
d_{\min} := \min_{p\neq j}\babs{y_p-y_j}\geqs \frac{2.4c_0e\pi }{\Omega }\Big(\frac{\sigma}{\sigma_{\infty, \min}\left(\vect A^{\top}\right)}\Big)^{\frac{1}{n}}
\end{equation}
with  $\frac{\sigma}{\sigma_{\infty, \min}\left(\vect A^{\top}\right)}\leqs 1$ and $\vect A$ being the amplitude matrix (\ref{equ:illuminationpattern1}). Moreover, for $\widehat \mu_t = \sum_{j=1}^n \widehat a_{j,t} \delta_{\widehat y_j}, \widehat y_j \in \mathcal O$ satisfying $||\mathcal F[\widehat \mu_t] - \vect Y_t||_{\infty}< \sigma, t=1, \cdots, T$, after reordering the $\widehat y_j$'s, we have 
	\begin{equation}
	\Big|\widehat y_j-y_j\Big|<\frac{d_{\min}}{2},
	\end{equation} 
	and 
	\begin{equation}
	\Big|\widehat y_j-y_j\Big| < \frac{C(n)}{\Omega}\mathrm{SRF}^{n-1}\frac{\sigma}{\sigma_{\infty, \min}\left(\vect A^{\top}\right)}, \quad 1\leqs j\leqs n,
	\end{equation}
	where $C(n)=\sqrt{\pi}n^2c_0^{n}e^{n}$ and $\mathrm{SRF} = \frac{\pi}{\Omega d_{\min}}$ is the super-resolution factor.
\end{thm}
\begin{proof}
\textbf{Step 1.} WLOG, we assume $\mathcal O = \left[-\frac{c_0 n\pi}{4\Omega}, \frac{c_0n\pi}{4\Omega} \right]$. Let $\mu_t = \sum_{j=1}^n a_{j,t} \delta_{y_j}$. Using the constraints
\[
\bnorm{\mathcal F[\widehat \mu_t] - \vect Y_t}_{\infty} < \sigma, \quad 1\leqs t\leqs T,
\]
we can derive that 
\begin{equation}\label{equ:multisupportresultequ2}
\max_{t=1, \cdots, T}\frac{1}{n+1}\sum_{j=1}^{n+1}|\mathcal F[\widehat \mu_t](\omega_j)-\mathcal F[\mu_t](\omega_j)|^2< 4 \sigma^2,
\end{equation}
where $\omega_j = (j-1)\tilde h-\Omega$ with $\tilde h = \frac{2\Omega}{c_0n}$. We consider 
\begin{equation*}
\left(\mathcal F [\widehat \mu_t](\omega_1), \mathcal F [\widehat \mu_t](\omega_{2}), \cdots,\mathcal F [\widehat \mu_t](\omega_{n+1})\right)^\top -\left(\mathcal F [\mu_t](\omega_1), \mathcal F [\mu_t](\omega_{2}), \cdots,\mathcal F [\mu_t](\omega_{n+1})\right)^\top =\widehat \Phi \widehat \alpha_t- \Phi \alpha_t, 
\end{equation*}
where $\widehat \alpha_t= (\widehat a_{1,t},\cdots, \widehat a_{n,t})^\top$, $\alpha_t=(a_{1,t}, \cdots, a_{n,t})^\top$ and 
\begin{equation*}
\widehat \Phi= \left(
\begin{array}{ccc}
e^{i\widehat y_1\omega_1}&\cdots& e^{i\widehat y_n\omega_1}\\
e^{i\widehat y_1\omega_{2}} &\cdots& e^{i\widehat y_n\omega_{2}}\\
\vdots&\vdots&\vdots\\
e^{i\widehat y_1\omega_{n+1}}&\cdots& e^{i\widehat y_n \omega_{n+1}}\\
\end{array}
\right), \quad \Phi=\left(
\begin{array}{ccc}
e^{i y_1\omega_1}&\cdots& e^{i y_n\omega_1}\\
e^{i y_1\omega_{2}} &\cdots& e^{i y_n\omega_{2}}\\
\vdots&\vdots&\vdots\\
e^{i y_1\omega_{n+1}}&\cdots& e^{i y_n \omega_{n+1}}\\
\end{array}
\right).
\end{equation*}
By (\ref{equ:multisupportresultequ2}), it is clear that 
\[
\max_{t=1, \cdots, T}\bnorm{\widehat \Phi \widehat \alpha_t- \Phi\alpha_t}_{2}<2\sqrt{n+1}\sigma.
\]

\textbf{Step 2.} Let $\theta_j = y_j \tilde h$ and $\widehat \theta_j = \widehat y_j \tilde h.$ Note that 
\begin{align}\label{equ:multiupperboundsupportlimithm1equ1}
\max_{t=1, \cdots, T}\bnorm{\widehat \Phi \widehat \alpha_t-\Phi \alpha_t}_2= \max_{t=1, \cdots, T}\bnorm{\widehat D \widehat \gamma_t -D \gamma_t}_2,
\end{align}
where $\widehat \gamma_{t}=\left(\widehat a_{1,t}e^{-i\widehat y_1\Omega},\cdots, \widehat a_{n,t}e^{-i\widehat y_n\Omega}\right)^\top, \gamma_t=\left(a_{1,t}e^{-iy_1\Omega},\cdots, a_{n,t}e^{-iy_n\Omega}\right)^\top$,\\
 $\widehat D=\big(\phi_{n}(e^{i \widehat \theta_1}),\cdots,\phi_{n}(e^{i \widehat \theta_n})\big)$ and $D=\big(\phi_{n}(e^{i \theta_1}),\cdots,\phi_{n}(e^{i \theta_n})\big)$. Thus,
\begin{equation}\label{equ:multisupportupperboundequ1}
\max_{t=1, \cdots, T}\bnorm{\widehat D \widehat \gamma_t-D \gamma_t}_2 < 2\sqrt{n+1}\sigma. 
\end{equation}
Since $\tilde h = \frac{2\Omega}{c_0n}$ and $\mathcal O = \left[-\frac{c_0 n\pi}{4\Omega}, \frac{c_0n\pi}{4\Omega} \right]$, we have $\theta_j, \widehat \theta_j \in \left[-\frac{\pi}{2}, \frac{\pi}{2}\right], j=1,\cdots, n$. Then we apply Theorem \ref{thm:multispaceapproxlowerbound2} to get
\begin{align}\label{equ:multiupperboundsupportlimithm1equ3}
\bnorm{\eta_{n,n}\left(e^{i \theta_1},\cdots,e^{i \theta_n},e^{i \widehat \theta_1},\cdots,e^{i \widehat \theta_n}\right)}_{\infty}<\frac{2^{n+1}\sqrt{n+1}\sigma}{\sigma_{\infty, \min}(B)}, 
\end{align}
where $\eta_{n,n}$ is defined by (\ref{equ:multieta}) and $B= \vect A^{\top}\text{diag}\left(e^{-iy_1\Omega},\cdots, e^{-iy_n\Omega}\right)$. Thus, we have
\begin{align}\label{equ:multisupportupperboundequ2}
\bnorm{\eta_{n,n}\left(e^{i \theta_1},\cdots,e^{i \theta_n},e^{i \widehat \theta_1},\cdots,e^{i \widehat \theta_n}\right)}_{\infty}<\frac{2^{n+1}\sqrt{n+1}\sigma}{\sigma_{\infty, \min}\left(\vect A^{\top}\right)}. 
\end{align}

\textbf{Step 3.}
Now, we apply Lemma \ref{lem:multistablemultiproduct0} to estimate $|\widehat \theta_j -\theta_j|$'s. For this purpose, let  $\epsilon = \frac{2\sqrt{n+1}\pi^{n}\sigma}{\sigma_{\infty, \min}\left(\vect A^{\top}\right)}$. It is clear that 
$\bnorm{\eta_{n ,n}}_{\infty}<(\frac{2}{\pi})^n\epsilon$ and we only need to check the following condition:
\begin{equation}\label{equ:multiupperboundsupportlimithm1equ4}
\theta_{\min}\geqs \Big(\frac{4\epsilon}{\lambda(n)}\Big)^{\frac{1}{n}}, \quad \mbox{or equivalently}\,\,\, (\theta_{\min})^n \geqs \frac{4\epsilon}{\lambda(n)}.
\end{equation}
Indeed, by $\theta_{\min} = \frac{2\Omega}{c_0n} d_{\min}$ and the separation condition (\ref{equ:onedsepaconditionsupport}), 
\begin{equation}\label{equ:multiupperboundsupportlimithm1equ-1}
\theta_{\min}\geqs  \frac{4.8\pi e}{n}\Big(\frac{\sigma}{\sigma_{\infty, \min}\left(\vect A^{\top}\right)}\Big)^{\frac{1}{n}}\geqs   \Big(\frac{8\sqrt{n+1}\pi^n}{\lambda(n)}\frac{\sigma}{\sigma_{\infty, \min}\left(\vect A^{\top}\right)}\Big)^{\frac{1}{n}}.
\end{equation}
Here, we have used Lemma \ref{lem:multisupportcalculate1} for deriving the last inequality. Then, we get (\ref{equ:multiupperboundsupportlimithm1equ4}). Therefore, we can apply Lemma \ref{lem:multistablemultiproduct0} to get that, after reordering $\widehat \theta_j$'s,
\begin{equation} \label{equ:onedmultiupperboundsupportlimithm1equ7}
\Big|\widehat \theta_{j}-\theta_j\Big|< \frac{\theta_{\min}}{2}, \text{ and } \Big|\widehat \theta_{j}-\theta_j\Big|< \frac{2^n\sqrt{n+1}\pi^{n}}{(n-2)!(\theta_{\min})^{n-1}}  \frac{\sigma}{\sigma_{\infty, \min}\left(\vect A^{\top}\right)},\ j=1,\cdots,n.
\end{equation}
Finally, we estimate $|\widehat y_j - y_j|$. Since $\babs{\widehat \theta_{j}-\theta_j}< \frac{\theta_{\min}}{2}$ and $\widehat y_j$'s, $y_j$'s are in $\mathcal O$, we have after reordering the $\widehat y_j$'s,
$$\babs{\widehat y_j-y_j}< \frac{d_{\min}}{2}.$$
On the other hand,  $\babs{\widehat y_j-y_j} = \frac{nc_0}{2\Omega}\Big|\widehat \theta_j -\theta_j\Big|$.  Together with (\ref{equ:onedmultiupperboundsupportlimithm1equ7}) and Lemma  \ref{lem:onedsupportuppercalculate1}, a direct calculation shows that
\begin{align*}
\Big|\widehat y_j-y_j\Big|< \frac{C(n)}{\Omega} \left(\frac{\pi}{\Omega d_{\min}}\right)^{n-1} \frac{\sigma}{\sigma_{\infty, \min}\left(\vect A^{\top}\right)}, 
\end{align*}
where $C(n)=n^2  \sqrt{\pi}c_0^{n}e^{n}$.
\end{proof}

Now, we prove Theorem \ref{thm:l0normrecovery0}.
\begin{proof}[Proof of Theorem \ref{thm:l0normrecovery0}]
By Theorem \ref{thm:onednumberbound}, any location set consisting of only $k<n$ elements cannot be a solution to (\ref{prob:l0minimization2}). On the other hand, the location set $\{y_1, \cdots, y_n\}$ is a already a feasible solution to (\ref{prob:l0minimization2}) since it contains only $n$ points and satisfies the measurement constraint. Combining the above arguments,  the solution to (\ref{prob:l0minimization2}) contains exactly $n$ points. Then by Theorem \ref{thm:onedsupportbound} we prove Theorem \ref{thm:l0normrecovery0}. 

\end{proof}

\subsection{Proof of Theorem \ref{thm:discretel0normrecovery0}}\label{section:proofdiscretel0normrecovery}
\begin{proof}
\textbf{Step 1.}
WLOG, we assume $\{x_j\}_{j=1}^N \subset \mathcal O = \left[-\frac{c_0 n\pi}{4\Omega}, \frac{c_0n\pi}{4\Omega} \right]$. Let $\widehat y_1, \cdots, \widehat y_k$ be points in the grid $\{x_j\}_{j=1}^N$ corresponding to the nonzero rows of $\mathbf{\widehat A}$. Denote 
\[
\widehat \mu_t=\sum_{j=1}^k \widehat a_{j,t}\delta_{\widehat y_j}, \quad \mu_t=\sum_{j=1}^n a_{j,t}\delta_{y_j}.
\]
Observe that the measurement constraint $\bnorm{\vect S- \vect \Phi \mathbf{\widehat A}}_{\mathrm{F}}<\sigma$ and noise constraint $\bnorm{\vect W}_{\mathrm{F}}<\sigma$ imply 
\[
\bnorm{\vect \Phi \mathbf{\widehat A} - \vect \Phi \mathbf{A}}_{\mathrm{F}}< 2\sigma
\]
or equivalently 
\begin{equation}\label{equ:proofdiscrete1}
\sum_{t=1}^{T}\bnorm{\left[\widehat \mu_t\right] - \left[\mu_t\right]}_{2}^2< 4\sigma^2,  
\end{equation}
where 
\begin{align*}
\left[\widehat \mu_t\right] &= \left(\mathcal F [\widehat \mu_t](\omega_1), \mathcal F [\widehat \mu_t](\omega_2), \cdots,\mathcal F [\widehat \mu_t](\omega_{M})\right)^\top,\\
\left[ \mu_t\right] &= \left(\mathcal F [\mu_t](\omega_1), \  \mathcal F  [\mu_t](\omega_2), \cdots,\mathcal F [\mu_t](\omega_{M})\right)^\top,
\end{align*}
where $\omega_1=-\Omega, \omega_2=-\Omega+h, \cdots, \omega_M=-\Omega+(M-1) h$ and $h=\frac{2 \Omega}{M-1}$. 

Similarly to the proof of Theorem \ref{thm:onednumberbound}, we shall first prove that if $k<n$, then for any $\widehat y_j\in \mathbb R, \widehat a_{j,t} \in \mathbb C, j=1,\cdots,k, t=1,\cdots,T$, (\ref{equ:proofdiscrete1}) does not hold. For ease of presentation, we fix $\widehat y_j, \widehat a_{j,t}$'s in the subsequent arguments.

\textbf{Step 2.} We first write
\begin{equation}\label{equ:proofdiscrete6}
M-1 =c_0n r+q,
\end{equation}
where $r\in \mathbb N$ and $0 \leq q<c_0n$. Since all the grid points $x_j\in \left[-\frac{c_0 n\pi}{4\Omega}, \frac{c_0n\pi}{4\Omega} \right], j=1,\cdots, N$, we have 
\begin{equation}\label{equ:proofdiscrete4}
\theta_j  :=y_j \frac{2 r \Omega}{M-1} \in\left[\frac{-\pi}{2}, \frac{\pi}{2}\right], \quad \widehat{\theta}_p :=\widehat{y}_p \frac{2 r \Omega}{M-1} \in\left[\frac{-\pi}{2}, \frac{\pi}{2}\right], 1 \leqs j\leqs n, \leqs p \leqs k.
\end{equation}

Using only the partial measurement at sampling points $\omega_{1+jr}=-\Omega+jrh, 0 \leqs j \leqs n-1$, we consider
$$
\begin{aligned}
& \left(\mathcal{F} [\widehat{\mu}_t]\left(\omega_1\right), \mathcal{F} [\widehat{\mu}_t]\left(\omega_{1+r}\right), \cdots, \mathcal{F} [\widehat{\mu}_t]\left(\omega_{1+(n-1) r}\right)\right)^{\top} \\
& \qquad -\left(\mathcal{F} [\mu_t]\left(\omega_1\right), \mathcal{F} [\mu_t]\left(\omega_{1+r}\right), \cdots, \mathcal{F} [\mu_t]\left(\omega_{1+(n-1) r}\right)\right)^{\top} =\widehat{\Phi} \widehat{a}_t-\Phi a_t
\end{aligned}
$$
where $\widehat{a}_t=\left(\widehat{a}_{1,t}, \cdots, \widehat{a}_{k,t}\right)^{\top}, a_t=\left(a_{1,t}, \cdots, a_{n,t}\right)^{\top}$ and
$$
\begin{gathered}
\widehat{\Phi}=\left(\begin{array}{ccc}
e^{i \widehat{y}_1 \omega_1} & \cdots & e^{i \widehat{y}_k \omega_1} \\
e^{i \widehat{y}_1 \omega_{1+r}} & \cdots & e^{i \widehat{y}_k \omega_{1+r}} \\
\vdots & \vdots & \vdots \\
e^{i \widehat{y}_1 \omega_{1+(n-1) r}} & \cdots & e^{i \widehat{y}_k \omega_{1+(n-1) r}}
\end{array}\right), 
\quad \Phi=\left(\begin{array}{ccc}
e^{i y_1 \omega_1} & \cdots & e^{i y_n \omega_1} \\
e^{i y_1 \omega_{1+r}} & \cdots & e^{i y_n \omega_{1+r}} \\
\vdots & \vdots & \vdots \\
e^{i y_1 \omega_{1+(n-1) r}} & \cdots & e^{i y_n \omega_{1+(n-1) r}}
\end{array}\right) .
\end{gathered}
$$
It is clear that
\begin{align*}
\sum_{t=1}^{T} \bnorm{[\hat{\mu}]-[\mu]}_{2}^2 \geqs \sum_{t=1}^{T} \bnorm{\widehat{\Phi} \widehat{a}_t-\Phi a_t}_{2}^2. 
\end{align*}
We shall prove that the following estimate holds:
\begin{equation}\label{equ:proofdiscrete2}
\sum_{t=1}^{T} \bnorm{\widehat \Phi \widehat \alpha_t- \Phi\alpha_t}_2^2> 4\sigma^2,
\end{equation}
 and consequently proves that (\ref{equ:proofdiscrete1}) doesn't hold. This proves $k\geqs n$ for the solution to (\ref{equ:jointsparsitydiscretemodel2}). On the other hand, since the matrix corresponding to $\mu_t$ is already a feasible solution to (\ref{equ:jointsparsitydiscretemodel2}), the sparsity-promoting nature of reconstruction (\ref{equ:jointsparsitydiscretemodel2}) gives $k=n$.

\textbf{Step 3.} Now we prove (\ref{equ:proofdiscrete2}). Following the decomposition:  
\begin{equation}\label{equ:multimatrixdecomposition1}
\begin{aligned}
&\widehat \Phi=\big(\phi_{n-1}(e^{i \widehat \theta_1}), \cdots,\phi_{n-1}(e^{i\widehat \theta_k} ) \big)\text{diag}(e^{-i\widehat y_1\Omega},\cdots,e^{-i\widehat y_k\Omega}),\\
&\Phi=\big(\phi_{n-1}(e^{i \theta_1}), \cdots,\phi_{n-1}(e^{i \theta_n})\big)\text{diag}(e^{-i y_1\Omega},\cdots,e^{-iy_n\Omega}),
\end{aligned}
\end{equation}
where $\phi_{n-1}(\cdot)$ is defined as in (\ref{equ:multiphiformula}) and $\theta_j, \widehat \theta_j$'s are defined as in (\ref{equ:proofdiscrete4}), we readily obtain that
\begin{align}\label{equ:proofdiscrete3} 
\sum_{t=1}^{T} \bnorm{\widehat \Phi \widehat \alpha_t-\Phi \alpha_t}_2^2= \sum_{t=1}^{T} \bnorm{\widehat D \widehat \gamma_t-D \gamma_t}_2^2,
\end{align}
where $\widehat \gamma_{t} = \left(\widehat a_{1,t}e^{-i\widehat y_1\Omega},\cdots, \widehat a_{k,t}e^{-i\widehat y_k\Omega}\right)^\top, \gamma_t=\left(a_{1,t}e^{-iy_1\Omega},\cdots, a_{n,t}e^{-iy_n\Omega}\right)^\top$, \\
$\widehat D=\big(\phi_{n-1}(e^{i \widehat \theta_1}), \cdots,\phi_{n-1}(e^{i \widehat \theta_k} ) \big)$ and $D=\big(\phi_{n-1}(e^{i \theta_1}), \cdots,\phi_{n-1}(e^{i \theta_n})\big)$.
We consider $\vect A$ in (\ref{equ:illuminationpattern1}) and denote $B= \vect A^{\top} \text{diag}\left(e^{-iy_1\Omega},\cdots, e^{-iy_n\Omega}\right)$. By (\ref{equ:proofdiscrete4}), we can apply Theorem \ref{thm:multispaceapprolowerbound3} to get 
\begin{equation}\label{equ:proofdiscrete5}
\sum_{t=1}^{T} \bnorm{\widehat D \widehat \gamma_t-D \gamma_t}_2^2\geqs   \frac{\sigma_{\min}(B)^2\xi(n-1)^2(\theta_{\min})^{2n-2}}{\pi^{2n-2}}, 
\end{equation}
where $\theta_{\min}=\min_{j\neq p}|\theta_j-\theta_p|$.
On the other hand, by the definition of $\sigma_{\min}(\cdot)$,
\[
\sigma_{\min}(B) = \sigma_{\min}\left(\vect A^{\top}\right).
\]
Together with (\ref{equ:proofdiscrete3}), we have 
\[
\sum_{t=1}^{T} \bnorm{\widehat \Phi \widehat \alpha_t-\Phi \alpha_t}_2^2 \geqs   \frac{\sigma_{\min}(\vect A^{\top})^2\xi(n-1)^2(\theta_{\min})^{2n-2}}{\pi^{2n-2}}. 
\]
On the other hand, recall that $d_{\min}= \min_{j\neq p}|y_j-y_p|$. Using relations (\ref{equ:proofdiscrete6}) and (\ref{equ:proofdiscrete4}), we have 
\[
\theta_{\min} = \frac{2 r \Omega}{M-1} d_{\min} \geqs \frac{2 r \Omega}{nc_0(r+1)}d_{\min} \geqs \frac{\Omega}{ nc_0}d_{\min}. 
\]
Then the separation condition (\ref{equ:mmvsepaconditionl0normrecovery}) and $\frac{\sigma}{\sigma_{\infty, \min}\left(\vect A^{\top}\right)} \leqs 1$ imply that
\begin{equation*}
\theta_{\min}\geqs \frac{4 e\pi }{n}\Big(\frac{\sigma}{\sigma_{\infty, \min}\left(\vect A^{\top}\right)}\Big)^{\frac{1}{n}} \geqs \frac{4 e\pi }{n}\Big(\frac{\sigma}{\sigma_{\infty, \min}\left(\vect A^{\top}\right)}\Big)^{\frac{1}{n-1}}> \pi \Big(\frac{2}{\xi(n-1)}\frac{\sigma}{\sigma_{\min}\left(\vect A^{\top}\right)}\Big)^{\frac{1}{n-1}},
\end{equation*}
where we have used Lemma \ref{lem:multinumbercalculate2} for deriving the last inequality. Therefore, 
\[
\sum_{t=1}^{T}\bnorm{\widehat \Phi \widehat \alpha_t-\Phi \alpha_t}_2^2 > 4\sigma^2,
\]
whence (\ref{equ:proofdiscrete2}) is proved.

\textbf{Step 4.} By above arguments, the solution of (\ref{prob:l0minimization2}) contains exactly $n$ points. Suppose that the solution is $\{\widehat y_1, \cdots, \widehat y_n\}$ and the sources in each snapshot is $\widehat \mu_{t} = \sum_{j=1}^n \widehat a_{j,t} \delta_{\widehat y_j}, t=1, \cdots, T,$ for certain $\widehat a_{j,t}$'s.  We now prove the stability of the location recovery. For $r$ defined in (\ref{equ:proofdiscrete6}), using only the partial measurement at $\omega_{1+jr}=-\Omega+jrh, 0 \leqs j \leqs n$, we consider
$$
\begin{aligned}
& \left(\mathcal{F} [\widehat{\mu}_t]\left(\omega_1\right), \mathcal{F} [\widehat{\mu}_t]\left(\omega_{1+r}\right), \cdots, \mathcal{F} [\widehat{\mu}_t]\left(\omega_{1+n r}\right)\right)^{\top} \\
& \qquad -\left(\mathcal{F} [\mu_t]\left(\omega_1\right), \mathcal{F} [\mu_t]\left(\omega_{1+r}\right), \cdots, \mathcal{F} [\mu_t]\left(\omega_{1+n r}\right)\right)^{\top} =\widehat{\Phi} \widehat{a}_t-\Phi a_t
\end{aligned}
$$
where $\widehat{a}_t=\left(\widehat{a}_{1,t}, \cdots, \widehat{a}_{n,t}\right)^{\top}, a_t=\left(a_{1,t}, \cdots, a_{n,t}\right)^{\top}$ and
$$
\begin{gathered}
\widehat{\Phi}=\left(\begin{array}{ccc}
e^{i \widehat{y}_1 \omega_1} & \cdots & e^{i \widehat{y}_n \omega_1} \\
e^{i \widehat{y}_1 \omega_{1+r}} & \cdots & e^{i \widehat{y}_n \omega_{1+r}} \\
\vdots & \vdots & \vdots \\
e^{i \widehat{y}_1 \omega_{1+n r}} & \cdots & e^{i \widehat{y}_n \omega_{1+n r}}
\end{array}\right), 
\quad \Phi=\left(\begin{array}{ccc}
e^{i y_1 \omega_1} & \cdots & e^{i y_n \omega_1} \\
e^{i y_1 \omega_{1+r}} & \cdots & e^{i y_n \omega_{1+r}} \\
\vdots & \vdots & \vdots \\
e^{i y_1 \omega_{1+n r}} & \cdots & e^{i y_n \omega_{1+n r}}
\end{array}\right).
\end{gathered}
$$
Here, we slightly abuse notation $\widehat \Phi$ and $\Phi$ for simplicity. By the constraint (\ref{equ:proofdiscrete1}), we thus have 
\[
\sum_{t=1}^{T}\bnorm{\widehat \Phi \widehat \alpha_t-\Phi \alpha_t}_2^2 < 4\sigma^2.
\]
Note that similarly to (\ref{equ:proofdiscrete5}), we still have 
\begin{align}
\sum_{t=1}^{T}\bnorm{\widehat \Phi \widehat \alpha_t-\Phi \alpha_t}_2^2= \sum_{t=1}^{T}\bnorm{\widehat D \widehat \gamma_t -D \gamma_t}_2^2,
\end{align}
where $\widehat \gamma_{t}=\left(\widehat a_{1,t}e^{-i\widehat y_1\Omega},\cdots, \widehat a_{n,t}e^{-i\widehat y_n\Omega}\right)^\top, \gamma_t=\left(a_{1,t}e^{-iy_1\Omega},\cdots, a_{n,t}e^{-iy_n\Omega}\right)^\top$,\\
 $\widehat D=\big(\phi_{n}(e^{i \widehat \theta_1}),\cdots,\phi_{n}(e^{i \widehat \theta_n})\big)$ and $D=\big(\phi_{n}(e^{i \theta_1}),\cdots,\phi_{n}(e^{i \theta_n})\big)$ with $\theta_j, \widehat \theta_j$'s defined as in (\ref{equ:proofdiscrete4}). This gives
\begin{equation}\label{equ:multisupportupperboundequ1}
\sum_{t=1}^{T}\bnorm{\widehat D \widehat \gamma_t-D \gamma_t}_2^2 < 4\sigma^2. 
\end{equation}
By (\ref{equ:proofdiscrete4}), we further apply Theorem \ref{thm:multispaceapproxlowerbound4} and get 
\begin{align*}
\bnorm{\eta_{n,n}\left(e^{i \theta_1},\cdots,e^{i \theta_n},e^{i \widehat \theta_1},\cdots,e^{i \widehat \theta_n}\right)}_{\infty}<\frac{2^{n+1}\sigma}{\sigma_{\infty, \min}(B)}, 
\end{align*}
where $\eta_{n,n}$ is defined by (\ref{equ:multieta}) and $B= \vect A^{\top}\text{diag}\left(e^{-iy_1\Omega},\cdots, e^{-iy_n\Omega}\right)$. Furthermore,
\begin{align}\label{equ:multisupportupperboundequ2}
\bnorm{\eta_{n,n}\left(e^{i \theta_1},\cdots,e^{i \theta_n},e^{i \widehat \theta_1},\cdots,e^{i \widehat \theta_n}\right)}_{\infty}<\frac{2^{n+1}\sigma}{\sigma_{\infty, \min}\left(\vect A^{\top}\right)}. 
\end{align}

\textbf{Step 5.}
Now, we apply Lemma \ref{lem:multistablemultiproduct0} to estimate $|\widehat \theta_j -\theta_j|$'s. For this purpose, let  $\epsilon = \frac{2\pi^{n}\sigma}{\sigma_{\infty, \min}\left(\vect A^{\top}\right)}$. It is clear that 
$\bnorm{\eta_{n ,n}}_{\infty}<(\frac{2}{\pi})^n\epsilon$ and we only need to check the following condition:
\begin{equation}\label{equ:proofdiscrete7}
\theta_{\min}\geqs \Big(\frac{4\epsilon}{\lambda(n)}\Big)^{\frac{1}{n}}, \quad \mbox{or equivalently}\,\,\, (\theta_{\min})^n \geqs \frac{4\epsilon}{\lambda(n)}.
\end{equation}
Indeed, by $\theta_{\min} = \frac{2\Omega}{c_0n} d_{\min}$ and the separation condition (\ref{equ:mmvsepaconditionl0normrecovery}), 
\begin{equation}\label{equ:proofdiscrete8}
\theta_{\min}\geqs  \frac{4\pi e}{n}\Big(\frac{\sigma}{\sigma_{\infty, \min}\left(\vect A^{\top}\right)}\Big)^{\frac{1}{n}}\geqs   \Big(\frac{8\pi^n}{\lambda(n)}\frac{\sigma}{\sigma_{\infty, \min}\left(\vect A^{\top}\right)}\Big)^{\frac{1}{n}}.
\end{equation}
Here, we have used Lemma \ref{lem:multisupportcalculate2} for deriving the last inequality. Then, we get (\ref{equ:proofdiscrete7}). Therefore, we can apply Lemma \ref{lem:multistablemultiproduct0} to get that, after reordering $\widehat \theta_j$'s,
\begin{equation} \label{equ:proofdiscrete8}
\Big|\widehat \theta_{j}-\theta_j\Big|< \frac{\theta_{\min}}{2}, \text{ and } \Big|\widehat \theta_{j}-\theta_j\Big|< \frac{2^n\pi^{n}}{(n-2)!(\theta_{\min})^{n-1}}  \frac{\sigma}{\sigma_{\infty, \min}\left(\vect A^{\top}\right)},\ j=1,\cdots,n.
\end{equation}
Finally, we estimate $|\widehat y_j - y_j|$. Since $\babs{\widehat \theta_{j}-\theta_j}< \frac{\theta_{\min}}{2}$ and $\widehat y_j$'s, $y_j$'s are in $\mathcal O$, we have after reordering the $\widehat y_j$'s,
$$\babs{\widehat y_j-y_j}< \frac{d_{\min}}{2}.$$
On the other hand,  $\babs{\widehat y_j-y_j} = \frac{nc_0}{2\Omega}\Big|\widehat \theta_j -\theta_j\Big|$.  Together with (\ref{equ:proofdiscrete8}) and Lemma  \ref{lem:onedsupportuppercalculate1}, a direct calculation shows that
\begin{align*}
\Big|\widehat y_j-y_j\Big|< \frac{C(n)}{\Omega} \left(\frac{\pi}{\Omega d_{\min}}\right)^{n-1} \frac{\sigma}{\sigma_{\infty, \min}\left(\vect A^{\top}\right)}, 
\end{align*}
where $C(n)=n^{\frac{3}{2}}  \sqrt{\pi}c_0^{n}e^{n}$.

\end{proof}

\section{Proof of Theorem \ref{thm:highdl0normrecovery0}}\label{section:proofhighdl0normrecovery}

\subsection{Projection lemmas}
We first introduce two lemmas from \cite{liu2021mathematicalhighd}. For a vector $\mathbf{v} \in \mathbb{R}^d$, we denote $\mathbf{v}^{\perp}$ the orthogonal complement space of the one-dimensional space spanned by $\mathbf{v}$. For $\mathbf{y} \in \mathbb{R}^d$ and a subspace $Q \subset \mathbb{R}^d$, we denote $\mathcal{P}_Q(\mathbf{y})$ the orthogonal projection of $\mathbf{y}$ onto $Q$.

\begin{lem}\label{lem:highdsupportproject3} 
For $n$ points $\mathbf{y}_j \in \mathbb{R}^d, d \geqs 2, n \geqs 2$ with minimum separation $d_{\min}:=\min _{p \neq j} \bnorm{\mathbf{y}_p- \mathbf{y}_j}_2$, let $\Delta=\frac{\pi}{8}\left(\frac{2}{(n+2)(n+1)}\right)^{\frac{1}{d-1}}$. There exist $n+1$ unit vectors $\vect v_q \in \mathbb{R}^d, q=1, \cdots, n+1$ such that $0 \leqs \mathbf{v}_p \cdot \mathbf{v}_j \leqs \cos 2 \Delta, p \neq j$ and
$$
\min _{p \neq j, 1 \leqs p, j \leqs n}\bnorm{\mathcal{P}_{\mathbf{v}_q^{\perp}}\left(\mathbf{y}_p\right)-\mathcal{P}_{\mathbf{v}_q^{\perp}}\left(\mathbf{y}_j\right)}_2 \geqs \frac{2 \Delta d_{\min }}{\pi}, \quad q=1, \cdots, n+1.
$$
\end{lem}
\begin{proof}
	See \cite[Lemma 3.2]{liu2021mathematicalhighd}. In particular, the correct value of $\Delta$ should be $\frac{\pi}{8}\left(\frac{2}{(n+2)(n+1)}\right)^{\frac{1}{d-1}}$.
\end{proof}

\begin{lem}\label{lem:highdsupportproject4}  
Let $d \geqslant 2$. For a vector $\mathbf{u} \in \mathbb{R}^d$, and two unit vectors $\mathbf{v}_1, \mathbf{v}_2 \in \mathbb{R}^d$ satisfying $0 \leqslant \mathbf{v}_1 \cdot \mathbf{v}_2 \leqslant \cos (\theta)$, we have
\[
\bnorm{\mathcal{P}_{\mathbf{v}_1^{\perp}}(\mathbf{u})}_2^2+\bnorm{\mathcal{P}_{\mathbf{v}_2^{\perp}}(\mathbf{u})}_2^2 \geqslant(1-\cos (\theta))\bnorm{\mathbf{u}}_2^2.
\]
\end{lem}
\begin{proof}
	See \cite[Lemma 3.3]{liu2021mathematicalhighd}. 
\end{proof}

\subsection{Proof of Theorem \ref{thm:highdl0normrecovery0}}
By Theorem \ref{thm:highdnumberbound}, under the separation condition (\ref{equ:highdsupportlimithm0equ0}), any solution to (\ref{prob:highdl0minimization2}) contains exactly $n$ points. The rest arguments of Theorem \ref{thm:highdl0normrecovery0} can be directly deduced from Theorem \ref{thm:highdsupportbound}. We next introduce and prove Theorems \ref{thm:highdnumberbound} and \ref{thm:highdsupportbound}.

We define
\begin{equation}\label{equ:defineofB}
B_\delta^d(\mathbf{x}):=\left\{\mathbf{y} \mid \mathbf{y} \in \mathbb{R}^d,\bnorm{\mathbf{y}-\mathbf{x}}_2<\delta\right\}.
\end{equation}

\begin{thm}\label{thm:highdnumberbound}
Suppose that the measurements $\vect Y_t$'s in (\ref{equ:highdmultimodelsetting1}) are generated from $n$ point sources located at $\vect y_j\in \mathbb R^d, j=1, \cdots, n,$ where $\vect y_j$'s are in a sphere $\mathcal O$ of radius $\frac{c_0n\pi}{4\Omega}$ with $c_0\geqs 1$ and satisfy
\begin{equation}\label{equ:highdnumberequ1}
	 \min_{p\neq j}\bnorm{\vect y_p-\vect y_j}_2\geqs \frac{2.4c_0e\pi 4^{d-1}((n+2)(n+1) / 2)^{\gamma(d-1)}}{\Omega}\Big(\frac{\sigma}{\sigma_{\infty, \min}\left(\vect A^{\top}\right)}\Big)^{\frac{1}{n}}
\end{equation}
with  $\frac{\sigma}{\sigma_{\infty, \min}\left(\vect A^{\top}\right)}\leqs 1$, $\vect A$ being the amplitude matrix (\ref{equ:highdilluminationpattern1}), and $\gamma(\cdot)$ defined by (\ref{equ:gammaformula1}). Then there is no $k<n$ locations $\mathbf {\widehat y}_j\in \mathbb R^d, j=1, \cdots,k$ such that there exists $\widehat \mu_t = \sum_{j=1}^k \widehat a_{j,t} \delta_{\mathbf{\widehat y}_j}$'s so that 
\[
\bnorm{\mathcal F[\widehat \mu_t] - \vect Y_t}_{\infty}< \sigma, \quad t=1, \cdots, T. 
\] 
\end{thm}
\begin{proof}
Without loss of generality, we assume $\mathcal O=  B_{\frac{c_0 n \pi}{4\Omega}}^{d}(\mathbf{0})$ in the proof. We prove the theorem by induction. The case when $d=1$ is exactly Theorem \ref{thm:onednumberbound}. Suppose Theorem \ref{thm:highdnumberbound} holds for the case when $d=\ell$, we now prove it for the case of $d=\ell+1$. Let the measurements $\mathbf{Y}_t$'s in (\ref{equ:highdmultimodelsetting1}) be generated by $\vect y_j\in B_{\frac{c_0 n \pi}{4 \Omega}}^{\ell+1}(\mathbf{0}), j=1,\cdots, n,$ satisfying the minimum separation condition
\begin{align}\label{equ:proofhighdnumberequ1}
	d_{\min }^{(\ell +1)} :=\min _{p \neq j, 1 \leqslant p, j \leqslant n}\bnorm{\mathbf{y}_p-\mathbf{y}_j}_2 \geqslant \frac{2.4 c_0 e \pi 4^{\ell}((n+2)(n+1)/2)^{\gamma(\ell)}}{\Omega}\left(\frac{\sigma}{\sigma_{\infty, \min}\left(\vect A^{\top}\right)}\right)^{\frac{1}{n}},
\end{align}
where $\gamma(\cdot)$ is defined by (\ref{equ:gammaformula1}). Let $\Delta = \frac{\pi}{8}\left(\frac{2}{(n+2)(n+1)}\right)^{\frac{1}{\ell}}$. By Lemma \ref{lem:highdsupportproject3}, there exist $n+1$ unit vectors $\mathbf{v}_q$ 's so that for each $q$,
$$
\min _{p \neq j}\left\|\mathcal{P}_{\mathbf{v}_q^{\perp}}\left(\mathbf{y}_p\right)-\mathcal{P}_{\mathbf{v}_q^{\perp}}\left(\mathbf{y}_j\right)\right\|_2 \geqslant d_{\min }^{(\ell)},
$$
where we define
$$
d_{\min }^{(\ell)}=\min _{p \neq j}\bnorm{\mathbf{y}_p-\mathbf{y}_j}_2 \frac{2 \Delta}{\pi}=\frac{d_{\min }^{(\ell +1)}}{4((n+2)(n+1) / 2)^{\frac{1}{\ell }}} .
$$
By (\ref{equ:proofhighdnumberequ1}) we have
\begin{align}
	&\min _{p \neq j} \bnorm{\mathcal{P}_{\mathbf{v}_{q}^{\perp}}\left(\mathbf{y}_p\right)-\mathcal{P}_{\mathbf{v}_{q}^{\perp}}\left(\mathbf{y}_j\right)}_2 \geqslant d_{\min }^{(\ell )} \nonumber \\
	&\qquad \geqslant\frac{2.4 c_0 e \pi 4^{\ell-1}((n+2)(n+1)/ 2)^{\gamma(\ell-1)}}{\Omega}\left(\frac{\sigma}{\sigma_{\infty, \min}\left(\vect A^{\top}\right)}\right)^{\frac{1}{n}} . \label{equ:proofhighdnumberequ2}
\end{align}
Suppose that there exists $k<n$ locations $\mathbf {\widehat y}_j\in \mathbb R^{\ell+1}, j=1, \cdots,k$ such that there exists $\widehat \mu_t = \sum_{j=1}^k \widehat a_{j,t} \delta_{\mathbf{\widehat y}_j}$'s so that 
\[
\bnorm{\mathcal F[\widehat \mu_t] - \vect Y_t}_{\infty}< \sigma, \quad t=1, \cdots, T. 
\] 
This implies for each $q$,
\begin{equation}\label{equ:proofhighdnumberequ3}
\bnorm{\sum_{j=1}^k\widehat a_{j,t} e^{ i \mathcal P_{\vect v_q^{\perp}}(\mathbf {\widehat y}_j)\cdot \vect \omega}    - \vect Y_t(\vect \omega)}_{\infty}< \sigma, \quad \vect \omega \in \vect v_q^{\perp}\cap B_{\Omega}^{\ell+1}(\vect 0), \quad t=1, \cdots, T. 
\end{equation}
Because the projected point sources $\delta_{\mathcal P_{\vect v_q^{\perp}}(\vect y_j)}$'s in the $\ell$-dimensional subspace $\vect v_q^{\perp}$ satisfy $\bnorm{\mathcal P_{\vect v_q^{\perp}}(\vect y_j)}_2\leqs \frac{c_0n \pi}{4\Omega}$ and are separated by a distance beyond (\ref{equ:proofhighdnumberequ2}), by the induction hypothesis, there should be no such $k<n$ positions $\mathcal P_{\vect v_q^{\perp}}(\mathbf {\widehat y}_j)$'s satisfying (\ref{equ:proofhighdnumberequ3}). This is a contradiction, which completes the proof. 
\end{proof}

\begin{thm}\label{thm:highdsupportbound}
	Suppose that the measurements $\vect Y_t$'s in (\ref{equ:highdmultimodelsetting1}) are generated from $n$ point sources located at $\vect y_j\in \mathbb R^d, j=1, \cdots, n,$ where $\vect y_j$'s are in a sphere $\mathcal O$ of radius $\frac{c_0n\pi}{4 \Omega}$ with $c_0\geqs 1$ and satisfy
	\begin{equation}\label{equ:highdsepaconditionsupport}
		d_{\min} := \min_{p\neq j}\bnorm{\vect y_p-\vect y_j}_2\geqs \frac{2.4c_0e\pi 4^{d-1}((n+2)(n+1) / 2)^{\gamma(d-1)}}{\Omega }\Big(\frac{\sigma}{\sigma_{\infty, \min}\left(\vect A^{\top}\right)}\Big)^{\frac{1}{n}}
	\end{equation}
	with $\frac{\sigma}{\sigma_{\infty, \min}\left(\vect A^{\top}\right)}\leqs 1$, $\vect A$ being the amplitude matrix (\ref{equ:highdilluminationpattern1}), and $\gamma(\cdot)$ defined by (\ref{equ:gammaformula1}). Moreover, for $\widehat \mu_t = \sum_{j=1}^n \widehat a_{j,t} \delta_{\mathbf{\widehat y}_j}, \mathbf{\widehat y}_j \in \mathcal O$ satisfying $\bnorm{\mathcal F[\widehat \mu_t] - \vect Y_t}_{\infty}< \sigma, t=1, \cdots, T$, after reordering the $\mathbf{\widehat y}_j$'s, we have 
	\begin{equation}
		\bnorm{\mathbf{\widehat y}_j-\vect y_j}_2<\frac{d_{\min}}{2},
	\end{equation} 
	and 
	\begin{equation}
		\bnorm{\mathbf{\widehat y}_j-\vect y_j}_2< \frac{C(d, n)}{\Omega}\mathrm{SRF}^{n-1}\frac{\sigma}{\sigma_{\infty, \min}\left(\vect A^{\top}\right)}, \quad 1\leqs j\leqs n,
	\end{equation}
	where 
\begin{equation}\label{equ:highdsupportboundequ1}	
C(d, n)= \sqrt{\pi}n^2c_0^{n}e^{n}4^{(d-1)n}\left(\frac{(n+2)(n+1)}{2}\right)^{\gamma(d-1)n}
\end{equation} 
and $\mathrm{SRF} = \frac{\pi}{\Omega d_{\min}}$ is the super-resolution factor.
\end{thm}	
\begin{proof}
	Without loss of generality, we assume $\mathcal O=  B_{\frac{c_0 n \pi}{4\Omega}}^{d}(\mathbf{0})$ in the proof. We prove the theorem by induction. The case when $d=1$ is exactly Theorem \ref{thm:onedsupportbound}. Suppose Theorem \ref{thm:highdsupportbound} holds for the case when $d=\ell$, we now prove it for the case of $d=\ell+1$. The beginning arguments are similar to those in the proof of Theorem \ref{thm:highdnumberbound}. In particular, let the measurements $\mathbf{Y}_t$'s in (\ref{equ:highdmultimodelsetting1}) be generated by $\mathbf{y}_j \in B_{\frac{c_0 n \pi}{4\Omega}}^{\ell+1}(\mathbf{0}), j=1,\cdots, n$, satisfying the minimum separation condition 
	\begin{align}\label{equ:proofhighdsupportequ1}
		d_{\min }^{(\ell +1)} :=\min _{p \neq j, 1 \leqslant p, j \leqslant n}\bnorm{\mathbf{y}_p-\mathbf{y}_j}_2 \geqslant \frac{2.4 c_0 e \pi 4^{\ell}((n+2)(n+1) / 2)^{\gamma(\ell)}}{\Omega}\left(\frac{\sigma}{\sigma_{\infty, \min}\left(\vect A^{\top}\right)}\right)^{\frac{1}{n}},
	\end{align}
	where $\gamma(\cdot)$ is defined by (\ref{equ:gammaformula1}). Let $\Delta = \frac{\pi}{8}\left(\frac{2}{(n+2)(n+1)}\right)^{\frac{1}{\ell}}$. By Lemma \ref{lem:highdsupportproject3}, there exist $n+1$ unit vectors $\mathbf{v}_q$ 's so that $0 \leqslant \mathbf{v}_p \cdot \mathbf{v}_j \leqslant \cos 2 \Delta, 1 \leqslant p<j \leqslant n$, and for each $q$, we have \begin{align}
		&\min _{p \neq j} \bnorm{\mathcal{P}_{\mathbf{v}_{q}^{\perp}}\left(\mathbf{y}_p\right)-\mathcal{P}_{\mathbf{v}_{q}^{\perp}}\left(\mathbf{y}_j\right)}_2 \geqslant d_{\min }^{(\ell )} \nonumber \\
		&\qquad \geqslant\frac{2.4 c_0 e \pi 4^{\ell-1}((n+2)(n+1) / 2)^{\gamma(\ell-1)}}{\Omega}\left(\frac{\sigma}{\sigma_{\infty, \min}\left(\vect A^{\top}\right)}\right)^{\frac{1}{n}} , \label{equ:proofhighdsupportequ2}
	\end{align}
where 
\[
d_{\min }^{(\ell)}=\min _{p \neq j}\bnorm{\mathbf{y}_p-\mathbf{y}_j}_2 \frac{2 \Delta}{\pi}=\frac{d_{\min }^{(\ell +1)}}{4((n+2)(n+1) / 2)^{\frac{1}{\ell }}} .
\]
Now for each $q$, consider the projected measure $\sum_{j=1}^n a_j \delta_{\mathcal{P}_{\vect v_q^{\perp}}\left(\mathbf{y}_j\right)}$ in the $\ell$-dimensional subspace $\mathbf{v}_q^{\perp}$ and the associated measurements $\mathbf{Y}_t(\boldsymbol{\omega}), \boldsymbol{\omega} \in \mathbf{v}_q^{\perp}\cap B_{\Omega}^{\ell+1}(\vect 0), t=1,\cdots, T$. It is clear that $\bnorm{\mathcal{P}_{\mathbf{v}_q^{\perp}}\left(\mathbf{y}_j\right)}_2<\frac{c_0n \pi}{ 4\Omega}$ and the separation condition (\ref{equ:highdsepaconditionsupport}) is satisfied for $d=\ell$ because of (\ref{equ:proofhighdsupportequ2}). On the other hand, 
	\[
	\bnorm{\mathcal F[\widehat \mu_t] - \vect Y_t}_{\infty}< \sigma, \quad t=1, \cdots, T,
	\]
	implies that 
	\[
	\bnorm{\mathcal F\left[\sum_{j=1}^n \widehat{a}_{j,t} \delta_{\mathcal{P}_{\mathbf{v}_{q}^{\perp}}\left(\widehat{\mathbf{y}}_j\right)}\right](\vect \omega) - \vect Y_t(\vect \omega)}_{\infty}< \sigma,\quad  \vect \omega \in \vect v_{q}^{\perp}\cap B_{\Omega}^{\ell+1}(\vect 0), \quad t=1, \cdots, T.
	\]
 Using the assumption that Theorem \ref{thm:highdsupportbound} holds for the case when $d=\ell$, we can conclude that for each $q$, we have a permutation $\tau_q$ of $\{1, \ldots, n\}$ so that
	\begin{equation}\label{equ:proofhighdsupportequ3}
	\bnorm{\mathcal{P}_{\mathbf{v}_q^{\perp}}\left(\widehat{\mathbf{y}}_{\tau_q(j)}\right)-\mathcal{P}_{\mathbf{v}_q}\left(\mathbf{y}_j\right)}_2 \leqslant \frac{C(\ell, n)}{\Omega}\left(\frac{\pi}{d_{\min }^{(\ell)} \Omega}\right)^{n-1} \frac{\sigma}{\sigma_{\infty, \min}\left(\vect A^{\top}\right)}, \quad 1 \leqslant j \leqslant n .
	\end{equation}
	Note that, for fixed $j$ in (\ref{equ:proofhighdsupportequ3}), we have $n+1$ different $\tau_q(j)$ 's, while $\widehat{\mathbf{y}}_p$ 's take at most $n$ values. Therefore, by the pigeonhole principle, for each fixed $\mathbf{y}_j$, we can find two different $q$ 's, say, $q_1$ and $q_2$, such that $\widehat{\mathbf{y}}_{\tau_{q_1}(j)}=\widehat{\mathbf{y}}_{\tau_{q_2}(j)}=\widehat{\mathbf{y}}_{p_j}$ for some $p_j$. Since $0 \leqslant \mathbf{v}_{q_1} \cdot \mathbf{v}_{q_2} \leqslant \cos 2 \Delta$, we can apply Lemma \ref{lem:highdsupportproject4} to get
	$$
	\bnorm{\widehat{\mathbf{y}}_{p_j}-\mathbf{y}_j}_2 \leqslant \frac{\sqrt{2}}{\sqrt{1-\cos (2 \Delta)}} \frac{C(\ell, n)}{\Omega}\left(\frac{\pi}{d_{\min }^{(\ell)} \Omega}\right)^{n-1} \frac{\sigma}{\sigma_{\infty, \min}\left(\vect A^{\top}\right)}, \quad 1 \leqslant j \leqslant n .
	$$
	Using the inequality $1-\cos 2 \Delta \geqslant \frac{8}{\pi^2} \Delta^2 \geqslant \frac{1}{8}\left(\frac{2}{(n+2)(n+1)}\right)^{\frac{2}{\ell}}$, we further obtain
	\begin{equation}\label{equ:proofhighdsupportequ4}
	\bnorm{\widehat{\mathbf{y}}_{p_j}-\mathbf{y}_j}_2 \leqslant \frac{4((n+2)(n+1)/2)^{\frac{1}{\ell}} C(\ell, n)}{\Omega}\left(\frac{\pi}{d_{\min }^{(\ell)} \Omega}\right)^{n-1}\frac{\sigma}{\sigma_{\infty, \min}\left(\vect A^{\top}\right)}, \quad 1 \leqslant j \leqslant n.
	\end{equation}
	We next claim that
	$$
	\bnorm{\widehat{\mathbf{y}}_{p_j}-\mathbf{y}_j}_2<\frac{d_{\min }^{(\ell+1)}}{2} .
	$$
	Indeed, by direct calculation, we can verify that
	$$
	\begin{aligned}
		4\left(\frac{(n+ 2)(n+1)}{2}\right)^{\frac{1}{\ell}}& C(\ell, n)\left(\frac{1}{2.4 c_0 e 4^{\ell-1}\left((n+2)(n+1)/2\right)^{\gamma(\ell -1)}}\right)^{n-1} \\
		& <\frac{1}{2} 2.4 c_0 \pi e 4^{\ell}\left(\frac{(n+2)(n+1)}{2}\right)^{\gamma(\ell)} , 
	\end{aligned}
	$$
	where $C(\ell, n)$ is defined as in (\ref{equ:highdsupportboundequ1}). On the other hand, (\ref{equ:proofhighdsupportequ2}) yields that
	$$
	\left(\frac{\pi}{d_{\min }^{(\ell)} \Omega}\right)^{n-1}\frac{\sigma}{\sigma_{\infty, \min}\left(\vect A^{\top}\right)}  \leqslant\left(\frac{1}{2.4 c_0 e 4^{\ell -1}((n+2)(n+1)/2)^{\gamma(\ell -1)}}\right)^{n-1} \left(\frac{\sigma}{\sigma_{\infty, \min}\left(\vect A^{\top}\right)}\right)^{\frac{1}{n}} .
	$$
	Therefore we have
	$$
	\begin{aligned}
		& \frac{4((n+2)(n+1)/2)^{\frac{1}{\ell}} C(\ell, n)}{\Omega}\left(\frac{\pi}{d_{\min }^{(\ell)} \Omega}\right)^{n-1}\frac{\sigma}{\sigma_{\infty, \min}\left(\vect A^{\top}\right)} \\
		&\qquad  <\frac{1}{2} \frac{2.4 c_0 \pi e 4^{\ell}((n+2)(n+1)/2)^{\gamma(\ell)}}{\Omega}\left(\frac{\sigma}{\sigma_{\infty, \min}\left(\vect A^{\top}\right)}\right)^{\frac{1}{n}} .
	\end{aligned}
	$$
	The claim follows by combining the above inequality with (\ref{equ:proofhighdsupportequ4}) and (\ref{equ:proofhighdsupportequ1}). So far, we have proved that for each $\mathbf{y}_j$, there exists a point $\widehat{\mathbf{y}}_{p_j}$ so that $	\bnorm{\mathbf{\widehat y}_{p_j} - \vect y_j}_2< \frac{d_{\min}^{(\ell+1)}}{2}$. Thus we can reorder the index so that 
	\[
	\bnorm{\mathbf{\widehat y}_j - \vect y_j}_2< \frac{d_{\min}^{(\ell+1)}}{2}, \quad j =1, \cdots, n.
	\]
	Moreover, we have
	$$
	\begin{aligned}
	\bnorm{\widehat{\mathbf{y}}_j-\mathbf{y}_j}_2 \leqslant \frac{\left(4((n+2)(n+1)/2)^{\frac{1}{\ell}}\right)^{n} C(\ell, n)}{\Omega}\left(\frac{\pi}{d_{\min }^{(\ell+1)} \Omega}\right)^{n-1} \frac{\sigma}{\sigma_{\infty, \min}\left(\vect A^{\top}\right)}, \quad 1 \leqslant j \leqslant n,
	\end{aligned}
	$$
	which follows from (\ref{equ:proofhighdsupportequ4}) and the equation that $d_{\min }^{(\ell)}=\frac{d_{\min }^{(\ell+1)}}{4((n+2)(n+1)/2)^{\frac{1}{\ell}}}$. This completes our induction argument and concludes the proof of the theorem.
\end{proof}

\section*{Acknowledgements}
The authors would like to thank the anonymous reviewers for their valuable suggestions for improving the paper. The work of PL was supported by Swiss National Science Foundation grant number 200021--200307. The work of SY was supported by the National Research Foundation of Korea under grant number NRF-2020R1C1C1A01010882.

\appendix 
\section{Proof of Proposition \ref{prop:multisupportlowerboundthm1}} \label{section:proofofsupportlowerbound}
We first introduce a lemma that was derived in \cite{liu2021theorylse}.
\begin{lem}{\label{lem:multiinvervandermonde}}
	Let $t_1, \cdots, t_k$ be $k$ different real numbers and let $t$ be a real number. We have
	\[
	\left(D_k(k-1)^{-1}\phi_{k-1}(t)\right)_{j}=\Pi_{1\leqs q\leqs k,q\neq j}\frac{t- t_q}{t_j- t_q},
	\]
	where $D_k(k-1):=  \big(\phi_{k-1}(t_1),\cdots,\phi_{k-1}(t_k)\big)$ and $\phi_{k-1}(\cdot)$ is defined by (\ref{equ:multiphiformula}). 
\end{lem}

\medskip
The proof of Proposition \ref{prop:multisupportlowerboundthm1} is divided into two steps.  
\begin{proof}
\textbf{Step 1.} Let $\tau$ be the one in (\ref{equ:multisupportlowerboundsepadis2}), $y_1=-\tau, y_2=-2\tau,\cdots, y_n=-n\tau$ and $\widehat y_1 =0, \widehat y_{2} = \tau, \cdots, \widehat y_{n} = (n-1)\tau$. Denote
\[
\mu_t = \sum_{j=1}^n a_{j,t} \delta_{y_j}, \quad  \widehat \mu_t = \sum_{j=1}^{n} \widehat a_{j,t} \delta_{\widehat y_j}, \quad t=1,\cdots, T. 
\]  
We aim to prove
\begin{equation}\label{equ:multisupportlowerboundthm0equ3}
\max_{x\in [-1, 1]} \babs{\mathcal F\left[\widehat \mu_t - \mu_t\right](\Omega x)}< \sigma. \quad x\in[-1,1], \quad t=1,\cdots, T.
\end{equation}
By Taylor expansion, 
\begin{align}
&\mathcal F \left[\widehat \mu_t - \mu_t\right](\Omega x)=\sum_{k=0}^{\infty}\big(Q_{k}\left(\widehat \mu_t\right)- Q_k\left(\mu_t\right)\big)\frac{(ix)^k}{k!},
\end{align} 
where $Q_{k}(\mu_t)=\sum_{j=1}^{n}a_{j,t} (\Omega y_{j})^k, \ Q_k(\widehat \mu_t) = \sum_{j=1}^{n}\widehat a_{j,t} (\Omega \widehat y_{j})^k$. Note that $Q_{k}(\widehat \mu_t) = Q_{k}(\mu_t), k=1, \cdots, n-1$, is a linear system that can be rewritten as
\begin{equation}
\big(\phi_{n-1}(\Omega \widehat y_1), \cdots, \phi_{n-1}(\Omega \widehat y_{n})\big)\widehat \alpha_t = \big(\phi_{n-1}(\Omega y_1), \cdots, \phi_{n-1}(\Omega y_{n})\big) \alpha_t,
\end{equation}
where $\widehat \alpha_t = (\widehat a_{1,t}, \cdots, \widehat a_{n, t})^T$ and $\alpha_t = (a_{1,t}, \cdots, a_{n,t})^T$. Since $\phi_{n-1}(\Omega \widehat y_j), 1\leqs j \leqs n,$ are linearly independent, for each $t$, we can find some $\widehat a_{j,t}$'s so that $Q_{k}(\widehat \mu_t) = Q_{k}(\mu_t), k=0,\cdots,n-1$. We next estimate $Q_{k}(\widehat \mu_t)$ and $Q_{k}(\mu_t)$. 

\textbf{Step 2.}
Under the scenario of Step 1, we obtain
\[
\widehat \alpha_t = \big(\phi_{n-1}(\Omega \widehat y_1), \cdots, \phi_{n-1}(\Omega \widehat y_{n})\big)^{-1}\big(\phi_{n-1}(\Omega y_1), \cdots, \phi_{n-1}(\Omega y_{n})\big) \alpha_t, \ 1\leqs t\leqs T.
\]
By Lemma \ref{lem:multiinvervandermonde}, we arrive at
\begin{align}\label{equ:proofonedmultiillulowerequ1}
\max_{p=1, \cdots, n}\Big|\widehat a_{p,t}\Big|\leqs& \max_{p=1, \cdots, n} \sum_{j=1}^n \Big|\Pi_{q=1, q\neq p}^{n}\frac{y_j - \widehat y_q}{ \widehat y_p- \widehat y_q} a_{j,t}\Big|\leqs \max_{p,j=1, \cdots, n}\Big|\Pi_{q=1, q\neq p}^{n}\frac{y_j - \widehat y_q}{ \widehat y_p- \widehat y_q}\Big|  \bnorm{\vect A}_{1}\nonumber \\
\leqs& \frac{(n+1)(n+2)\cdots(2n -1)}{\zeta(n)} \bnorm{\vect A}_{1} = \frac{(2n -1)!}{n! \zeta(n)} \bnorm{\vect A}_{1}  \quad  \Big(\text{$\zeta(\cdot)$ is defined in (\ref{equ:multizetaxiformula1})}\Big)\nonumber \\
\leqs& \frac{ e 2^{3n-2}}{\pi^{\frac{3}{2}} (n-1)}\bnorm{\vect A}_{1}  \quad \Big(\text{using inequality (\ref{equ:stirlingformula})}\Big). \nonumber 
\end{align}
Further, we have $\sum_{j=1}^{n}|\widehat a_{j,t}|\leqs  \frac{ e 2^{3n-2}}{\pi^{\frac{3}{2}} (n-1)}  n  \bnorm{\vect A}_{1}, 1\leqs t\leqs T$. Therefore, we have
\begin{align*}
\Big|\max_{x\in [-1, 1]}\mathcal F\left[\widehat \mu_t - \mu_t\right](\Omega x)\Big|\leqs &\sum_{k\geqs n} \frac{|Q_k(\mu_t)|+|Q_k(\widehat \mu_t)|}{k!}
\leqs  \sum_{k\geqs n} \left(\sum_{j=1}^n |a_{j,t}|+ |\widehat a_{j,t}|\right) \frac{(n\Omega \tau)^k}{k!}\\
\leqs & \Big(\frac{ e 2^{3n-2}}{\pi^{\frac{3}{2}} (n-1)}  n+1\Big)\bnorm{\vect A}_{1}\frac{(n\Omega \tau)^{n}}{n!}\sum_{k\geqs n}\frac{(n\Omega\tau)^{k-n}n!}{k!} \\
< & 1.06\Big(\frac{ e 2^{3n-2}}{\pi^{\frac{3}{2}} (n-1)}  n+1\Big)\bnorm{\vect A}_{1}\frac{(n\Omega \tau)^{n}}{n!} \quad \Big( \text{by (\ref{equ:multisupportlowerboundsepadis2}) we have $\Omega \tau < 0.05$}\Big)\\
< & 1.06\Big(\frac{ e 2^{3n-2}}{\pi^{\frac{3}{2}} (n-1)} n+1\Big)\bnorm{\vect A}_{1} \frac{e^n}{\sqrt{2\pi n} }(\Omega \tau)^{n} \quad \Big(\text{by (\ref{equ:stirlingformula})}\Big)\\
=& 1.06\Big(\frac{ e 2^{3n-2}}{\pi^{\frac{3}{2}} (n-1)} n+1\Big) \frac{e^n}{\sqrt{2\pi n} }(0.044)^{n}\sigma \quad \Big(\text{by (\ref{equ:multisupportlowerboundsepadis2})}\Big)\\
< & \sigma  \quad (\text{by Lemma \ref{lem:multisupportlowercalculate1}}).
\end{align*}
This completes the proof.
\end{proof}

\section{Some estimations}
\label{appendixb}
In this section, we present some estimations that are used in this paper. We first recall the following Stirling approximation of factorial
\begin{equation}\label{equ:stirlingformula}
\sqrt{2\pi} n^{n+\frac{1}{2}}e^{-n}\leqs n! \leqs e n^{n+\frac{1}{2}}e^{-n}.
\end{equation}

Then, we state the following results. 
\begin{lem}\label{lem:multinumbercalculate1}
Let $\xi(n-1)$ be defined as in (\ref{equ:multizetaxiformula1}). For $n\geqs 2$, we have 
\[
\Big(\frac{2\sqrt{n}n^{n-1}}{\xi(n-1)}\Big)^{\frac{1}{n-1}} < 4.8 e. 
\]	
\end{lem}
\begin{proof} For $n=2,3,4,5$, it is easy to check that the above inequality holds. Using (\ref{equ:stirlingformula}), we have for odd $n\geqs 7$,
\begin{align*}
&\Big(\frac{2\sqrt{n}n^{n-1}}{\xi(n-1)}\Big)^{\frac{1}{n-1}}= \Big(\frac{2\sqrt{n}n^{n-1}}{\frac{(\frac{n-3}{2}!)^2}{4}}\Big)^{\frac{1}{n-1}}\leqs \bparenths{\frac{8\sqrt{n}n^{n-1}}{\pi(\frac{n-3}{2})^{n-2}e^{-(n-3)}}}^{\frac{1}{n-1}}\\
=& \bparenths{\frac{8\sqrt{n}}{\pi(\frac{n-3}{2})^{-1}e^{2}}}^{\frac{1}{n-1}}\frac{2en}{n-3}\\
<& 4.8 e,
\end{align*}
and for even $n\geqs 6$,
\begin{align*}
&\Big(\frac{2\sqrt{n}n^{n-1}}{\xi(n-1)}\Big)^{\frac{1}{n-1}}= \Big(\frac{2\sqrt{n}n^{n-1}}{\frac{(\frac{n-2}{2})!(\frac{n-4}{2})!}{4}}\Big)^{\frac{1}{n-1}}\leqs \bparenths{\frac{8\sqrt{n}n^{n-1}}{\pi(\frac{n-2}{2})^{\frac{n-1}{2}}(\frac{n-4}{2})^{\frac{n-3}{2}}e^{-(n-3)}}}^{\frac{1}{n-1}}\\
=&  \bparenths{\frac{8\sqrt{n}}{\pi(\frac{n-4}{2})^{-1}e^{2}}}^{\frac{1}{n-1}}\frac{2en}{\sqrt{n-2}\sqrt{n-4}}\\
<& 4.8e.  
\end{align*}
This completes the proof. 
\end{proof}

\begin{lem}\label{lem:multinumbercalculate2}
Let $\xi(n-1)$ be defined as in (\ref{equ:multizetaxiformula1}). For $n\geqs 2$, we have 
\[
\Big(\frac{2n^{n-1}}{\xi(n-1)}\Big)^{\frac{1}{n-1}} < 4 e. 
\]	
\end{lem}
\begin{proof}
For $n=2,3,4,5,6$, it is easy to check that the above inequality holds. Using (\ref{equ:stirlingformula}), we have for odd $n\geqs 7$,
\begin{align*}
&\Big(\frac{2n^{n-1}}{\xi(n-1)}\Big)^{\frac{1}{n-1}}= \Big(\frac{2 n^{n-1}}{\frac{(\frac{n-3}{2}!)^2}{4}}\Big)^{\frac{1}{n-1}}\leqs \bparenths{\frac{8n^{n-1}}{\pi(\frac{n-3}{2})^{n-2}e^{-(n-3)}}}^{\frac{1}{n-1}}\\
=& \bparenths{\frac{8}{\pi(\frac{n-3}{2})^{-1}e^{2}}}^{\frac{1}{n-1}}\frac{2en}{n-3}\\
<& 4 e,
\end{align*}
and for even $n\geqs 8$,
\begin{align*}
&\Big(\frac{2n^{n-1}}{\xi(n-1)}\Big)^{\frac{1}{n-1}}= \Big(\frac{2n^{n-1}}{\frac{(\frac{n-2}{2})!(\frac{n-4}{2})!}{4}}\Big)^{\frac{1}{n-1}}\leqs \bparenths{\frac{8n^{n-1}}{\pi(\frac{n-2}{2})^{\frac{n-1}{2}}(\frac{n-4}{2})^{\frac{n-3}{2}}e^{-(n-3)}}}^{\frac{1}{n-1}}\\
=&  \bparenths{\frac{8}{\pi(\frac{n-4}{2})^{-1}e^{2}}}^{\frac{1}{n-1}}\frac{2en}{\sqrt{n-2}\sqrt{n-4}}\\
<& 4e.  
\end{align*}
This completes the proof.
\end{proof}

\begin{lem}\label{lem:multisupportcalculate1}
Let $\lambda(n)$ be defined as in (\ref{equ:lambda1}). For $n\geqs 2$, we have 
\begin{equation*}
\Big(\frac{8\sqrt{n+1}n^n}{\lambda(n)}\Big)^{\frac{1}{n}}< 4.8e.
\end{equation*}	
\end{lem}
\begin{proof} For $n=2,3,\cdots,14$, it is easy to check that the above inequality holds. Using (\ref{equ:stirlingformula}), we have for even $n\geqs 16$,
\begin{align*}
&\Big(\frac{8\sqrt{n+1}n^{n}}{\xi(n-2)}\Big)^{\frac{1}{n}}= \Big(\frac{8\sqrt{n+1}n^{n}}{(\frac{n-4}{2}!)^2/4}\Big)^{\frac{1}{n}}\leqs \bparenths{\frac{32\sqrt{n+1}n^{n}}{\pi(\frac{n-4}{2})^{n-3}e^{-(n-4)}}}^{\frac{1}{n}}\\
=& \bparenths{\frac{32\sqrt{n+1}}{\pi(\frac{n-4}{2})^{-3}e^{4}}}^{\frac{1}{n}}\frac{2en}{n-4}\\
<& 4.8e,
\end{align*}
and for odd $n\geqs 15$,
\begin{align*}
&\Big(\frac{8\sqrt{n+1}n^{n}}{\xi(n-2)}\Big)^{\frac{1}{n}}= \Big(\frac{8\sqrt{n+1}n^{n}}{\frac{(\frac{n-3}{2})!(\frac{n-5}{2})!}{4}}\Big)^{\frac{1}{n}}\leqs \bparenths{\frac{32\sqrt{n+1}n^{n}}{\pi(\frac{n-3}{2})^{\frac{n-2}{2}}(\frac{n-5}{2})^{\frac{n-4}{2}}e^{-(n-4)}}}^{\frac{1}{n}}\\
=&  \bparenths{\frac{32\sqrt{n+1}}{\pi(\frac{n-3}{2})^{-1}(\frac{n-5}{2})^{-2}e^{4}}}^{\frac{1}{n}}\frac{2en}{\sqrt{n-3}\sqrt{n-5}}\\
<& 4.8e.  
\end{align*}
This completes the proof. 
\end{proof}

\begin{lem}\label{lem:multisupportcalculate2}
Let $\lambda(n)$ be defined as in (\ref{equ:lambda1}). For $n\geqs 2$, we have 
\begin{equation*}
\Big(\frac{8n^n}{\lambda(n)}\Big)^{\frac{1}{n}}< 4e.
\end{equation*}	
\end{lem}
\begin{proof} For $n=2,3,\cdots,14$, it is easy to check that the above inequality holds. Using (\ref{equ:stirlingformula}), we have for even $n\geqs 16$,
\begin{align*}
&\Big(\frac{8n^{n}}{\xi(n-2)}\Big)^{\frac{1}{n}}= \Big(\frac{8n^{n}}{(\frac{n-4}{2}!)^2/4}\Big)^{\frac{1}{n}}\leqs \bparenths{\frac{32n^{n}}{\pi(\frac{n-4}{2})^{n-3}e^{-(n-4)}}}^{\frac{1}{n}}\\
=& \bparenths{\frac{32}{\pi(\frac{n-4}{2})^{-3}e^{4}}}^{\frac{1}{n}}\frac{2en}{n-4}\\
<& 4e,
\end{align*}
and for odd $n\geqs 15$,
\begin{align*}
&\Big(\frac{8n^{n}}{\xi(n-2)}\Big)^{\frac{1}{n}}= \Big(\frac{8n^{n}}{\frac{(\frac{n-3}{2})!(\frac{n-5}{2})!}{4}}\Big)^{\frac{1}{n}}\leqs \bparenths{\frac{32n^{n}}{\pi(\frac{n-3}{2})^{\frac{n-2}{2}}(\frac{n-5}{2})^{\frac{n-4}{2}}e^{-(n-4)}}}^{\frac{1}{n}}\\
=&  \bparenths{\frac{32}{\pi(\frac{n-3}{2})^{-1}(\frac{n-5}{2})^{-2}e^{4}}}^{\frac{1}{n}}\frac{2en}{\sqrt{n-3}\sqrt{n-5}}\\
<& 4e.  
\end{align*}
This completes the proof. 
\end{proof}

\begin{lem}\label{lem:onedsupportuppercalculate1}
For $n\geqs 2$, 
\begin{align}
\frac{n^n\sqrt{n+1}}{(n-2)!}&\leqs \frac{n^2}{\sqrt{\pi}}e^{n}, \label{equ:supportstabilitycal1}\\
\frac{n^n}{(n-2)!} &\leqs \frac{n^{\frac{3}{2}}}{\sqrt{\pi}}e^{n}. \label{equ:supportstabilitycal2}
\end{align}
\end{lem}
\begin{proof}
For $2\leqslant n < 4$, (\ref{equ:supportstabilitycal1}) can be verified directly. For $n\geqs 4$, by (\ref{equ:stirlingformula}) we have 
\begin{align*}
\frac{n^n\sqrt{n+1}}{(n-2)!}\leqs&  \frac{n^n\sqrt{n+1}}{\sqrt{2\pi}(n-2)^{n-2}\sqrt{n-2}e^{-(n-2)}}\\
=& n^2\frac{\sqrt{n+1}}{\sqrt{2\pi(n-2)}}\left(1+\frac{2}{n-2}\right)^{n-2} e^{n-2}\\
\leqs& \frac{n^2}{\sqrt{\pi}}e^{n}. 
\end{align*}
The inequality (\ref{equ:supportstabilitycal2}) can be proved similarly. 
\end{proof}

\begin{lem}\label{lem:multisupportlowercalculate1}
For $n\geqs 2$, 
\[
1.06\Big(\frac{ e 2^{3n-2}}{\pi^{\frac{3}{2}} (n-1)} n+1\Big) \frac{e^n}{\sqrt{2\pi n} }(0.044)^{n}<1.
\]
\end{lem}
\begin{proof}
As $0.044e<0.12$, we  only need to prove 
\[
1.06\Big(\frac{ e 2^{-2}0.96^n}{\pi^{\frac{3}{2}}(n-1)} n+0.12^n \Big) \frac{1}{\sqrt{2\pi n} }<1.
\]
Since $\frac{0.96\times 29}{28}<1$, it is not hard to see that when $n\geqs 29$, the LHS of the above inequality decreases as $n$ increases. For $2\leqs n\leqs 29$, we have checked numerically that the inequality holds. 
\end{proof}


\bibliographystyle{plain}
\bibliography{references_final}	
	
\end{document}